\documentclass[runningheads]{llncs}
\usepackage{latexsym,stmaryrd}
\usepackage{amssymb, amsmath}
\usepackage{dsfont}
\usepackage{nicematrix}

\usepackage{thmtools}
\usepackage{thm-restate}

\newcommand{\grammar}{\ensuremath{\Lambda}}
\newcommand{\dee}{\mathrm{d}}

\newcommand{\cB}{\mathcal{B}}
\newcommand{\cC}{\mathcal{C}}

\newcommand{\cE}{\mathcal{E}}
\newcommand{\cF}{\mathcal{F}}

\newcommand{\cH}{\mathcal{H}}
\newcommand{\cI}{\mathcal{I}}

\newcommand{\cL}{\mathcal{L}}
\newcommand{\cM}{\mathcal{M}}

\newcommand{\cO}{\mathcal{O}}
\newcommand{\cP}{\mathcal{P}}

\newcommand{\cX}{\mathcal{X}}

\begin{document}

\title{A behavioural pseudometric for continuous-time Markov processes}
\author{Linan Chen \inst{1} \and
Florence Clerc \inst{2} \and
Prakash Panangaden \inst{1}}

\authorrunning{Chen et al.}

\institute{McGill University \and
Heriot-Watt University}

\maketitle

\begin{abstract}
In this work,  we generalize the concept of bisimulation metric in order to metrize the behaviour of continuous-time processes. Similarly to what is done for discrete-time systems, we follow two approaches and show that they coincide: as a fixpoint of a functional and through a real-valued logic.

The whole discrete-time approach relies entirely on the step-based dynamics: the process jumps from state to state. We define a behavioural pseudometric for processes that evolve continuously through time, such as Brownian motion or involve jumps or both.
\end{abstract}

\section{Introduction}

Bisimulation is a concept that captures behavioural equivalence of states in a variety
  of types of transition systems.  It has been widely studied in a discrete-time setting
  where the notion of a step is fundamental.  An important and especially useful further notion is that of bisimulation metric which quantifies ``how similar two states are''.

Most of the theoretical work that exists is on discrete time but a growing part of what computer science allows us to do is in real-time: robotics, self-driving cars, online machine-learning etc. A common solution is to discretize time, however it is well-known that this can lead to errors that are hopefully small but that may accumulate over time and lead to vastly different outcomes. For that reason, it is important to have a continuous-time way of quantifying the error made.

Bisimulation~\cite{Milner80,Park81,Sangiorgi09} is a fundamental concept in the theory of transition systems capturing a strong notion of behavioural equivalence.  The extension to probabilistic systems is due to Larsen and Skou~\cite{Larsen91}; henceforth we will simply say ``bisimulation'' instead of ``probabilistic bisimulation''.  Bisimulation has been studied for discrete-time systems where transitions happen as steps, both on discrete~\cite{Larsen91} and continuous state spaces~\cite{Blute97,Desharnais98,Desharnais02}.  In all these types of systems, a crucial ingredient of the definition of bisimulation is the ability to talk about \emph{the next step}.  This notion of bisimulation is characterized by a modal logic~\cite{Larsen91} even when the state space is continuous~\cite{Desharnais98}.

Some work had previously been done in what are called continuous-time systems, see for example~\cite{Baier08}, but even in so-called continuous-time Markov chains there is a discrete notion of time \emph{step}; it is only that there is a real-valued duration associated with each state that leads to people calling such systems continuous-time. They are often called ``jump processes'' in the mathematical literature (see, for example, \cite{Rogers00a,Whitt02}), a phrase that better captures the true nature of such processes.  Metrics and equivalences for such processes were studied by Gupta et al.~\cite{Gupta04,Gupta06}.

The processes we consider have continuous state spaces and are governed by a continuous-time evolution, a paradigmatic example is Brownian motion.  When approximating such processes by discrete-time processes,  entirely new phenomena and difficulties manifest themselves in this procedure.  For example, even the basic properties of trajectories of Brownian motion are vastly more complicated than the counterparts of a random walk.  Basic concepts like ``the time at which a process exits a given subset of the state space'' becomes intricate to define.  Notions like ``matching transition steps'' are no longer applicable as the notion of ``step'' does not make sense. 

In~\cite{Chen19a,Chen20,Chen23},  we proposed different notions of behavioural
equivalences on continuous-time processes. We showed that there were several possible extensions of the notion of bisimulation to continuous time and that the continuous-time notions needed to involve trajectories in order to be meaningful.  There were significant mathematical challenges in even proving that an equivalence relation existed.  For example,  obstacles occurred in establishing measurability of various functions and sets, due to the inability to countably generate the relevant $\sigma$-algebras.  Those papers left completely open the question of defining a suitable pseudometric analogue, a concept that would be more useful in practice than an equivalence relation.

Previous work on discrete-time Markov processes by Desharnais et al.~\cite{Desharnais99b,Desharnais04}
extended the modal logic characterizing bisimulation to a real-valued logic that allowed
to not only state if two states were ``behaviourally equivalent'' but, more interestingly, how similarly they behaved.  This shifts the notion from a qualitative notion (an
equivalence) to a quantitative one (a pseudometric).

Other work also on discrete-time Markov processes by van Breugel et
al.~\cite{vanBreugel05a} introduced a slightly different real-valued logic and compared
the corresponding pseudometric to another pseudometric obtained as a terminal coalgebra of a
carefully crafted functor. We also mention in this connexion the work by Ferns et al. on Markov Decision Processes and the connexion between bisimulation and optimal value functions \cite{Ferns11}.

In this work, we are looking to extend the notion of bisimulation metric to a behavioural pseudometric on continuous-time processes.
Very broadly speaking, we are following a familiar path from equivalences to logics to metrics.  However, it is necessary for us to redevelop the framework and the mathematical techniques from scratch. 
Indeed, a very important aspect in discrete-time is the fact that the process is a jump process, ``hopping'' from state to state. This limitation also applies to continuous-time Markov chains. In our case, we want to cover processes that evolve through time. A standard example would be Brownian motion or other diffusion processes (often described by stochastic differential equations). As one will see throughout this work, there are new mathematical challenges that need to be overcome.
This means that the similarity between the pre-existing work on discrete-time and our generalization to continuous-time is only at the highest level of abstraction.

\paragraph*{Outline of the paper:}
The first two sections after the introduction are background. We will start by recalling some mathematical notions in Section \ref{sec:background}, introducing the continuous-time processes that we will be studying in Section \ref{sec:CTsystems}. A very brief overview of bisimulation metrics in discrete time can also be found in Appendix \ref{app:DT-metrics}. In Section \ref{sec:gen-CT}, we will introduce a functional $\cF$ and define a pseudometric $\overline{\delta}$ using this functional. We will also show that the pseudometric $\overline{\delta}$ is a fixpoint of $\cF$. In Section \ref{sec:logics}, we will show that this pseudometric is characterized by a real-valued logic.
We will further emphasize the novelty of this work wrt discrete time and summarize the obstacles that we had to overcome in Section \ref{sec:wrapup}. 
We will provide some examples in Section \ref{sec:example}. 
Finally we will discuss the limitations of our approach  and how it relates to previous works in Section~\ref{sec:conclusion}.

\section{Mathematical background}
\label{sec:background}

We assume the reader to be familiar with basic measure
theory and topology. Nevertheless we provide a brief review of the relevant notions and theorems. Let us start with clarifying a few notations on integrals: Given a measurable space $X$ equipped with a measure $\mu$ and a measurable function $f: X \to \mathbb{R}$, 
we can write either $\int f~ \dee \mu$ or $\int f(x) ~ \mu (\dee x)$ interchangeably. The second notation will be especially useful when considering a Markov kernel $P_t$ for some $t \geq 0$ and $x \in X$: $\int f(y) ~P_t(x, \dee y) = \int f ~ \dee P_t(x)$.

All the proofs for this Section can be found in Appendix \ref{app:proof-sec-background}.

   \subsection{Lower semi-continuity}

\begin{definition}
Given a topological space $X$, a function $f: X\to \mathbb{R}$ is \emph{lower semi-continuous} if for every $x_0 \in X$, $\liminf_{x \rightarrow x_0} f(x) \geq f(x_0)$. This condition is equivalent to the following one: for any $y \in \mathbb{R}$, the set $f^{-1}((y, + \infty)) = \{ x ~|~ f(x) > y\}$ is open in $X$.
\end{definition}

Let $X$ be a metric space. A function $f: X \to \mathbb{R}$ is lower semi-continuous,  if and only if $f$ is the limit of an increasing sequence of real-valued continuous functions on $X$. Details can be found in the Appendix \ref{app:semi-cont}.

    \subsection{Couplings}
\label{sec:coupling}

\begin{definition}
Let $(X, \Sigma_X, P)$ and $(Y, \Sigma_Y, Q)$ be two probability spaces.  Then a \emph{coupling} $\gamma$ of $P$ and $Q$ is a probability distribution on $(X \times Y, \Sigma_X \otimes \Sigma_Y)$ such that for every  $B_X \in \Sigma_X$, $\gamma(B_X \times Y) = P(B_X)$ and for every $B_Y \in \Sigma_Y$, $\gamma(X \times B_Y) = Q(B_Y)$ ($P, Q$ are called the \emph{marginals} of $\gamma$).  We write $\Gamma(P, Q)$ for the set of couplings of $P$ and~$Q$.
\end{definition}

\begin{restatable}{lemma}{lemmacouplingscompact}
\label{lemma:couplings-compact}
Given two probability measures $P$ and $Q$ on Polish spaces $X$ and $Y$ respectively,  the set of
couplings $\Gamma (P, Q)$ is compact under the topology of weak convergence. 
\end{restatable}

   \subsection{Optimal transport theory}

A lot of this work is based on optimal transport theory. This whole subsection is based on \cite{Villani08} and will be adapted to our framework.

Consider a Polish space $\cX$ and a lower semi-continuous cost function $c: \cX \times \cX \to [0,1]$ such that for every $x \in \cX$, $c(x,x) = 0$.

For every two probability distributions $\mu$ and $\nu$ on $\cX$, we write $W(c)(\mu, \nu)$ for the optimal transport cost from $\mu$ to $\nu$.
Adapting Theorem 5.10(iii) of \cite{Villani08} to our framework (see Remark \ref{rem:dual-pseudometric} in Appendix \ref{app:transport}),  we get the following statement for the Kantorovich duality:
\[W(c)(\mu, \nu) =  \min_{\gamma \in \Gamma(\mu, \nu)} \int c ~ \dee \gamma
= \max_{h \in \cH(c)} \left| \int h ~ \dee \mu - \int h ~\dee \nu  \right| \]
where $ \cH(c) = \{ h: \cX \to [0,1] ~|~ \forall x,y ~~ |h(x) - h(y)| \leq c(x,y) \}$.
   
\begin{restatable}{lemma}{lemmaoptimaltransportpseudometric}
\label{lemma:optimal-transport-pseudometric}
If the cost function $c$ is a 1-bounded  pseudometric on $\cX$, then $W(c)$ is a 1-bounded pseudometric on the space of probability distributions on $\cX$.
\end{restatable}

We will later need the following technical lemma. Theorem 5.20 of \cite{Villani08} states that a sequence $W(c_k)(P_k, Q_k)$ converges to $W(c)(P, Q)$ if $c_k$ uniformly converges to $c$ and $P_k$ and $Q_k$ converge weakly to $P$ and $Q$ respectively. Uniform convergence in the cost function may be too strong a condition for us, but the following lemma is enough for what we need.

\begin{restatable}{lemma}{lemmawassersteinlimit}
\label{lemma:wasserstein-limit}
Consider a Polish space $\cX$ and a cost function $c: \cX \times \cX \to [0,1]$ such that there exists an increasing ($c_{k+1} \geq c_k$ for every $k$) sequence of continuous cost functions $c_k: \cX \times \cX \to [0,1]$ that converges to $c$ pointwise.  Then, given two probability distributions $P$ and $Q$ on $\cX$,
\[ \lim_{k \rightarrow \infty} W(c_k)(P,Q) = W(c)(P, Q). \]
\end{restatable}

\section{Background on continuous-time Markov processes}
\label{sec:CTsystems}

This work focuses on continuous-time processes that are honest (without loss of mass over time) and with additional regularity conditions.  In order to define what we mean by continuous-time Markov processes here, we first define Feller-Dynkin processes.  Much of this material is adapted from \cite{Rogers00a} and we use
their notations.  Another useful source is \cite{Bobrowski05}.  

Let $E$ be a locally compact, Hausdorff
space with a countable base.  We also equip the set $E$ with its Borel $\sigma$-algebra $\mathcal{B}(E)$, denoted $\cE$.  The previous topological hypotheses also imply that $E$ is $\sigma$-compact and Polish (see corollary IX.57 in \cite{Bourbaki89b}).
We will denote $\Delta$ for the 1-bounded metric that generates the topology making $E$ Polish.

\begin{definition}
A \emph{semigroup} of operators on any Banach space $X$ is a family
of linear continuous (bounded) operators $\cP_t: X \to X$ indexed by
$t\in\mathbb{R}_{\geq 0}$ such that
\[ \forall s,t \geq 0, \cP_s \circ \cP_t = \cP_{s+t} \qquad \text{(semigroup property)}\]
and
\[ \cP_0 = I \qquad \text{(the identity)}.  \]
\end{definition}
\begin{definition}
For $X$ a Banach space, we say that a semigroup $\cP_t:X\to X$ is \emph{strongly continuous} if 
\[ \forall x\in X, \lim_{t\downarrow 0}\| \cP_t x - x \| \to 0.  \]
\end{definition}

What the semigroup property expresses is that we do not need to understand the past (what
happens before time $t$) in order to compute the future (what happens after some
additional time $s$, so at time $t+s$) as long as we know the present (at time
$t$). 

We say that a continuous real-valued function $f$ on $E$ ``vanishes at
infinity'' if for every $\varepsilon > 0$ there is a compact subset
$K \subseteq E$ such that for every $x\in E\setminus K$, we have
$|f(x)| \leq \varepsilon$.  To give an intuition, if $E$ is the real line, this means that
$\lim_{x \rightarrow \pm \infty} f(x) = 0$.  The space $C_0(E)$ of continuous real-valued functions
that vanish at infinity is a Banach space with the $\sup$ norm.   
\begin{definition}
A \emph{Feller-Dynkin (FD) semigroup} is a strongly continuous
semigroup $(\hat{P}_t)_{t \geq 0}$ of linear operators on $C_0(E)$ satisfying the
additional condition:  
\[\forall t \geq 0 ~~~ \forall f \in C_0(E) \text{, if }~~ 0 \leq f \leq 1 \text{, then }~~ 0 \leq \hat{P}_t f \leq 1\]
\end{definition}

The Riesz representation theorem can be found as Theorem II.80.3 of \cite{Rogers00a}.  From it, we can derive the following important proposition which relates these FD-semigroups with Markov
kernels (see Appendix \ref{app:CTsystems} for the details).  This allows one to see the connection with familiar
probabilistic transition systems.  

\begin{proposition}
\label{prop:Riesz-use}
Given an FD-semigroup $(\hat{P}_t)_{t \geq 0}$ on $C_0(E)$, it is possible to define a
unique family of sub-Markov 
kernels $(P_t)_{t \geq 0} : E \times \mathcal{E} \to [0,1]$ such that for
all $t \geq 0$ and $f \in C_0(E)$, 
\[ \hat{P}_t f(x) = \int f(y) P_t(x, \dee y).  \]
\end{proposition}

Given a time $t$ and a state $x$, we will often write $P_t(x)$ for the measure $P_t(x, \cdot)$ on $E$. Note that since $E$ is Polish, then $P_t(x)$ is tight.

\begin{definition}
\label{def:honest}
A process described by the FD-semigroup $(\hat{P}_t)_{t \geq 0}$ is \emph{honest} if for every $x \in E$ and every time $t \geq 0$, $P_t(x, E) = 1$.
\end{definition}

Worded differently, a process is honest if there is no loss of mass over time. A standard example of an honest process is Brownian motion. We refer the reader to Appendix \ref{app:BM} for an introduction to Brownian motion.

  \subsection{Observables}

In previous sections,  we defined Feller-Dynkin processes.  In order to bring the processes more in line
with the kind of transition systems that have hitherto been studied in the computer
science literature, we also equip the state space $E$ with an additional continuous function
$obs: E \to [0,1]$. 
One should think of it as the interface between the process and the
user (or an external observer): external observers won't see the exact state in which the process
is at a given time, but they will see the associated \emph{observables}.  What could be a real-life example is the depth at which a diver goes: while the diver does not
know precisely his location underwater, at least his watch is giving him the depth at
which he is.

Note that the condition on the observable is a major difference from our previous work \cite{Chen19a,Chen20} since we used a countable set of atomic propositions $AP$ and $obs$ was a discrete function $E \to 2^{AP}$. 

\begin{definition}
\label{def:CTMarkov}
In this study,  a \emph{Continuous-time Markov process} (abbreviated CTMP) is an honest FD-semigroup on $C_0(E)$  equipped with a continuous function $obs: E \to [0,1]$ and that satisfies the following additional property: if a sequence $(x_n)_{n \in \mathbb{N}}$ converges to $x$ in $E$, then for every $t$, the sequence of measures $(P_t(x_n))_{n \in \mathbb{N}}$ weakly converges to the measure $P_t(x)$.
\end{definition}

\begin{remark}
Some properties could be relaxed. For instance, in some cases, a non honest process could be made into a CTMP by adding a state $\partial$. Another hypothesis that could be relaxed is the one on $obs$ by imposing some stronger conditions on the FD-process.
\end{remark}

%
%
%
%
%

\section{Generalizing to continuous-time through a functional}
\label{sec:gen-CT}

We start by defining a behavioural pseudometric on our CTMPs by defining a functional $\cF$ on the lattice of 1-bounded pseudometrics.  As we will see, unlike in the discrete-time case, it is not possible to apply the Banach fixpoint theorem and get a fixpoint metric a priori: instead we need to construct a candidate and then show that it is a fixpoint of our functional. More specifically, the idea is to iteratively apply our functional to a metric and then consider the supremum of the sequence of pseudometrics. Doing so requires to first restrict the scope of our functional $\cF$.

\subsection{Lattices}

At the core of this construction is the definition of a functional on the lattice of 1-bounded pseudometrics.

Let $\cM$ be the lattice of 1-bounded pseudometrics on the state space $E$ equipped with the
order $\leq$ defined as: $m_1 \leq m_2$ if and only if for every $(x,y)$,
$m_1(x,y) \leq m_2(x,y)$.  We can define a sublattice $\cP$ of $\cM$ by restricting to pseudometrics that are lower semi-continuous (wrt the original topology $\cO$ on $E$ generated by the metric $\Delta$ making the space $E$ Polish). We will further require to define the sublattice $\cC$ which is the set of pseudometrics $m \in \cM$ on the state space $E$ such that the topology  generated by $m$ on $E$ is a subtopology of the original topology $\cO$, \emph{i.e.}\  $m$ is a continuous function $E \times E \to [0,1]$.

We have the following inclusion: $\cC \subset \cP \subset \cM$.

\begin{remark}
\label{rem:lattices}
One has to be careful here. The topology $\cO$ on $E$ is generated by the 1-bounded metric $\Delta$, and hence $\Delta$ is in $\cC$. However, we can define many pseudometrics that are not related to $\cO$.  As an example, the discrete pseudometric\footnote{The discrete pseudometric is defined as $m(x, y) = 1$ if $x \neq y$ and $m(x,x) = 0$} on the real line is not related to the usual topology on $\mathbb{R}$. 
\end{remark}

\subsection{Defining our functional}
\label{sec:def-functional}

Throughout the rest of the paper, $(P_t)_{t>=0}$ is the family of Markov kernels associated with a CTMP.
Given a discount factor $0 < c< 1$,  we define the functional $\cF_c : \cP \to \cM$ as follows: for every pseudometric $m \in \cP$ and every two states $x,y$,
\[ \cF_c(m)(x,y) = \sup_{t \geq 0} c^t W(m) (P_t(x), P_t(y)).\]
$ \cF_c(m)(x,y) $ compares all the distributions $P_t(x)$ and $P_t(y)$  through transport theory and takes their supremum. 

There are several remarks to make on this definition.
First, we can only define $\cF_c (m)$ if $m$ is lower semi-continuous since we are using optimal transport theory which is why the domain of $\cF_c$ is only $\cP$.

Additionally, even if $m$ is lower semi-continuous, $\cF_c(m)$ may not even be measurable which means that the range of $\cF_c$ is not the lattice $\cP$. At least,  Lemma \ref{lemma:optimal-transport-pseudometric} ensures that $\cF_c(m)$ is indeed in $\cM$, as a supremum of pseudometrics. 
This subtlety was not present in the work on continuous-time Markov chains in ~\cite{Gupta04,Gupta06}.

Second, we will use the Kantorovich duality throughout this work. It only holds for probability measures, and that is why we restrict this work to honest processes. 

As a direct consequence of the definition of $\cF_c$, we have that $\cF_c$  is monotone: if $m_1 \leq m_2$ in $\cP$, then $\cF_c(m_1) \leq \cF_c(m_2)$.

\begin{restatable}{lemma}{lemmaordermFc}
\label{lemma:order-m-Fc}
For every pseudometric $m$ in $\cP$,  discount factor $0 < c < 1$ and pair of states~$x,y$,
\[ m(x,y) \leq \cF_c(m)(x,y). \]
\end{restatable}

The proof can be found in Appendix \ref{app:gen-CT}.

\subsection{When restricted to continuous pseudometrics}
\label{sec:func-restrict-cont}

We wish to iteratively apply $\cF_c$ in order to construct a fixpoint (in a similar fashion to the proof of the Knaster-Tarski theorem). While $\cF_c(m)$ is a pseudometric (for $m \in \cP$), there is no reason for it to be in $\cP$. This means that we cannot hastily apply $\cF_c$ iteratively to just any pseudometric in order to obtain a fixpoint.

However, if $m$ is a pseudometric which is continuous wrt the original topology, then so is $\cF_c(m)$.

\begin{restatable}{lemma}{lemmaFcont}
\label{lemma:Fc-cont-is-cont}
Consider a pseudometric $m \in \cC$.  Then the topology generated by $\cF_c(m)$ is a subtopology of the original topology $\cO$ for any discount factor $0< c< 1$. 
\end{restatable}

This is where we need that the discount factor $c<1$. The condition that $c< 1$ enables us to maintain continuity by allowing to bound the time interval we consider. Indeed, given $T> 0$, for any time $t \geq T$ and any $x, y \in E$, we know that $c^t W(m) (P_t(x), P_t(y)) \leq~c^T$. 

\begin{proof}
Since $c$ and $m$ are fixed throughout the proof, we will omit noting them and for instance write $\cF(x,y)$ instead of $\cF_c(m)(x,y)$ and $W(P_t(x), P_t(y))$ instead of $W(m)(P_t(x), P_t(y))$. We will also write $\Phi(t,x,y) = c^t W(m)(P_t(x), P_t(y))$, i.e. $\cF(x,y) = \sup_t \Phi(t,x,y)$.

It is enough to show that for a fixed state $x$, the map $y \mapsto \cF(x,y)$ is continuous.

Pick $\epsilon >0$ and a sequence of states $(y_n)_{n \in \mathbb{N}}$ converging to $y$.  We want to show that there exists $M$ such that for all $n \geq M$,
\begin{align}
\label{eq:limitF}
|\cF(x, y) - \cF(x, y_n)| & \leq \epsilon.
\end{align}

Pick $t$ such that $\cF(x,y) = \sup_s \Phi(s, x,y) \leq \Phi(t, x, y) + \epsilon / 4$, i.e.
\begin{align}
\label{eq:numbera}
| \Phi(t, x, y)  - \cF(x,y)| \leq  \epsilon / 4.
\end{align}

Recall that $P_t(y_n)$ converges weakly to $P_t(y)$ and hence we can apply Theorem 5.20 of \cite{Villani08} and get:
\[ \lim_{n \rightarrow \infty} W(P_t(x), P_t(y_n)) = W(P_t(x), P_t(y)). \]
This means that there exists $N'$ such that for all $n \geq N'$, 
\[ |W(P_t(x), P_t(y_n)) - W(P_t(x), P_t(y))| \leq \epsilon/4. \]
This further implies that for all $n \geq N'$, 
\begin{align}
\label{eq:numberb}
|\Phi(t, x, y_n) - \Phi(t, x, y)|
&  \leq c^t \epsilon/4  \leq \epsilon / 4.
\end{align}

In order to show (\ref{eq:limitF}), it is enough to show that there exists $N$ such that for every $n \geq N$,
\begin{align}
\label{eq:limitF-intermediary}
|\Phi(t, x, y_n) - \cF(x, y_n)| & \leq \epsilon / 2.
\end{align}
Indeed, in that case, $\forall n \geq \max\{N, N'\}$,  using Equations (\ref{eq:numbera}) and (\ref{eq:numberb}),
\begin{align*}
&|\cF(x, y) - \cF(x, y_n)|\\
& \leq |\cF(x, y) - \Phi(t, x, y)| + |\Phi(t, x, y) - \Phi(t, x, y_n)| + |\Phi(t, x, y_n) - \cF(x, y_n)| \\
& \leq \epsilon / 4 + \epsilon / 4 + \epsilon / 2 = \epsilon.
\end{align*}

So let us show (\ref{eq:limitF-intermediary}). Assume it is not the case: for all $N$, there exists $n \geq N$ such that $ |\Phi(t, x, y_n) - \cF(x, y_n)| > \epsilon / 2$, i.e.
\[ \Phi(t, x, y_n) + \epsilon / 2 <  \cF(x, y_n).  \]
Define the sequence $(N_k)_{k \in \mathbb{N}}$ by: $N_{-1} = -1$ and if $N_k$ is defined, define $N_{k+1}$ to be the smallest $n \geq N_k + 1$ such that $\Phi(t, x, y_n) + \epsilon / 2 <  \cF(x, y_n)$.
In particular for every $k \in \mathbb{N}$,  $\cF(x, y_{N_k})  > \epsilon / 2$. There exists $T$ such that for every $s \geq T$, $c^s < \epsilon/2$.  We thus have that
\[ \forall k \in \mathbb{N} \quad \cF(x, y_{N_k}) = \sup_{ 0 \leq s \leq T} \Phi(s, x, y_{N_k}).\]
Therefore for every $k \in \mathbb{N}$, there exists $s_k \in [0, T]$ such that
\begin{align}
\label{eq:numberc}
\cF(x, y_{N_k}) \leq \Phi(s_k, x, y_{N_k}) + \epsilon / 8.
\end{align}
We get a sequence $(s_k)_{k \in \mathbb{N}} \subset [0,T]$, and there is thus a subsequence $(t_k)_{k \in \mathbb{N}}$ converging to some $t' \in [0, T]$. There is a corresponding subsequence $(z_k)_{k \in \mathbb{N}}$ of the original sequence $(y_{N_k})_{k \in \mathbb{N}}$. Since $\lim_{n \rightarrow \infty} y_n = y$ ,  $\lim_{k \rightarrow \infty} z_k = y$. 

We constructed the sequence $(N_k)_{k \in \mathbb{N}}$ such that  $\Phi(t, x, y_{N_k}) + \epsilon / 2 <  \cF(x, y_{N_k})$. Hence by Equation (\ref{eq:numberc}),
\[  \Phi(t, x, z_k) + \epsilon / 2 < \cF (x, y_{N_k}) \leq \Phi(t_k, x, z_k) + \epsilon / 8,  \]
which means that by taking the limit $k \to \infty$,
\begin{align}
\label{eq:contr1}
 \Phi(t, x, y) + \epsilon / 2 &\leq \Phi(t', x, y) + \epsilon / 8.
\end{align}

At the start of this proof, we picked $t$ such that $\cF(x,y) = \sup_s \Phi(s, x,y) \leq \Phi(t, x, y) + \epsilon / 4$ which means that
\begin{align}
\label{eq:contr2}
 \Phi(t', x, y) &\leq \Phi(t, x,y) + \epsilon / 4 .
\end{align}
Equations (\ref{eq:contr1}) and (\ref{eq:contr2}) are incompatible which concludes the proof. \hfill $\square$
\end{proof}

\subsection{Defining our family of pseudometrics}
\label{sec:def-pseudometrics-F}

We are now finally able to iteratively apply our functional $\cF_c$ on continuous pseudometrics and thus construct a sequence of increasing pseudometrics and its limit.

Since $obs$ is a continuous function $E \to \mathbb{R}_{\geq 0}$ and by Lemma \ref{lemma:Fc-cont-is-cont}, we can define the sequence of pseudometrics in $\cC$ for each $0 < c< 1$:
\begin{align*}
\delta^c_0(x,y) & = |obs(x) - obs(y)|,\\
\delta^c_{n+1} &= \cF_c(\delta_n^c).
\end{align*}
By Lemma \ref{lemma:order-m-Fc}, for every two states $x$ and $y$, $\delta^c_{n+1}(x,y) \geq \delta^c_n(x,y)$.  Define the pseudometric $\overline{\delta}^c = \sup_n \delta^c_n$ (which is also a limit since the sequence is non-decreasing). 

As a direct consequence of Lemma \ref{lemma:sup-lower-semicontinuous},  the pseudometric $\overline{\delta}^c$ is lower semi-continuous and is thus in the lattice $\cP$ for any $0 < c < 1$.

\begin{remark}
\label{rem:tarski}
Note that the lattice $\cC$ is not complete which means that, although the metrics $\delta^c_n$ all belong to $\cC$,   $\overline{\delta}^c$ does not need to be in $\cC$.  For that reason, we cannot directly use the Knaster-Tarski theorem in this work.
\end{remark}

\subsection{Fixpoint}
\label{sec:metrics-fixpoints}

Even though we are not able to define the metric  $\overline{\delta}^c$ as a fixpoint directly, it is actually a fixpoint.

\begin{restatable}{theorem}{thmmfixpointF}
\label{thm:m-fixpointFc}
The pseudometric $\overline{\delta}^c$ is a fixpoint for $\cF_c$.
\end{restatable}

The full proof can be found in Appendix  \ref{app:gen-CT}. It relies on the use of Sion's minimax theorem on the function
\begin{align*}
\Xi : \Gamma (P_t(x), P_t(y)) \times Y & \to [0,1]\\
(\gamma, m) & \mapsto\int m ~ \dee \gamma.
\end{align*}
where $Y$ is the set of linear combinations of pseudometrics $\sum_{n \in \mathbb{N}} a_n \delta_n$ such that for every $n$, $a_n \geq 0$ (and finitely many are non-zero) and $\sum_{n \in \mathbb{N}} a_n = 1$. 

\begin{lemma}
\label{lemma:comparison-fixpointFc}
Consider a discount factor $0 < c < 1$ and a pseudometric $m$ in $\cP$ such that
\begin{enumerate}
\item $m$ is a fixpoint for $\cF_c$,
\item for every two states $x$ and $y$, $m(x,y) \geq |obs(x) - obs(y)|$,
\end{enumerate}
then $m \geq \overline{\delta}^c$.
\end{lemma}

The proof is done by showing that $m \geq \delta_n^c$ by induction on $n$.

%

 We even have the following characte\-rization of $\overline{\delta}^c$ using Lemma \ref{lemma:comparison-fixpointFc} and Theorem \ref{thm:m-fixpointFc}.
 
\begin{theorem}
The pseudometric  $\overline{\delta}^c$ is the least fixpoint of $\cF_c$ that is greater than the pseudometric $(x, y) \mapsto |obs(x) - obs(y)|$.
\end{theorem}

\section{Corresponding real-valued logic}
\label{sec:logics}

Similarly to what happens in the discrete-time setting, this behavioural pseudometric $\overline{\delta}^c$ can be characterized by a real-valued logic. This real-valued logic should be thought of as tests performed on the diffusion process, for instance ``what is the expected value of $obs$ after letting the process evolve for time $t$?'' and it generates a pseudometric on the state space by looking at how different the process performs on those tests starting from different positions.

\subsection{The logic}

\paragraph*{Definition of the  logic:}
The logic is defined inductively and is denoted $\grammar$:
\begin{align*}
f \in \grammar & := q~|~ obs ~|~ \min\{ f_1, f_2 \} ~|~ 1-f ~|~ f \ominus q ~|~ \langle t \rangle f
\end{align*}
for all $f_1, f_2, f \in \grammar$, $q \in [0,1] \cap \mathbb{Q}$ and $t \in \mathbb{Q}_{ \geq 0}$.

This logic closely resembles the ones introduced for discrete-time systems by Desharnais et al.~\cite{Desharnais99b,Desharnais04} and by van Breugel et al.~\cite{vanBreugel05a}. The key difference is the term $\langle t \rangle f$ which deals with continuous-time.

\paragraph*{Interpretation of the logics: }
We fix a discount factor $0 < c< 1$. The expressions in $\grammar$ are interpreted inductively as functions $E \to [0,1]$ as follows for a state $x \in E$:
\begin{align*}
q(x) & = q, \\
obs(x) & = obs(x), \\
(\min\{ f_1, f_2 \})(x) & = \min\{ f_1(x),  f_2(x) \}, \\
(1-f)(x) & = 1 - f(x), \\
(f \ominus q)(x) & = \max \{ 0, f(x) - q \}, \\
(\langle t \rangle f)(x) &= c^t ~ \int f(y) ~ P_t(x, \dee y) = c^t~ \left(\hat{P}_t f\right)(x).
\end{align*}
Whenever we want to emphasize the fact that the expressions are interpreted for a discount factor $0<c<1$, we will write $\grammar_c$.

\begin{remark}
Let us clarify what the difference is between an expression in $\grammar$ and its
interpretation. Expressions can be thought of as the
notation $+, ^2, \times$ etc. They don't carry much meaning by
themselves but one can then interpret them for a given set: $\mathbb{R},  C_0(E)$
(continuous functions $E \to \mathbb{R}$ that vanish at infinity) for instance. Combining
notations, one can write expressions that can then be interpreted
on a given set.
\end{remark}

From $\grammar$, we can also define the expression $f \oplus q  = 1 - ((1 - f) \ominus q)$
which is interpreted as a function $E \to [0,1]$ as  $(f \oplus q)(x)  = \min \{ 1, f(x) + q \}$.

\subsection{Definition of the pseudometric}

The pseudometric we derive from the logic $\grammar$ corresponds to how different the test results are when the process starts from $x$  compared to the case when it starts from $y$.

Given a fixed discount factor $0< c<1$, we can define the pseudometric $\lambda^c$:
\[ \lambda^c (x,y) = \sup_{f \in \grammar_c} (f (x) - f(y)) = \sup_{f \in \grammar_c} |f(x) - f(y)|.  \]
The latter equality holds since for every $f \in \grammar_c$, $\grammar_c$ also contains $1-f$.
%

\subsection{Comparison to the fixpoint metric}
This real-valued logic $\grammar_c$ is especially interesting as the corresponding pseudometric $\lambda^c$ matches the fixpoint pseudometric $\overline{\delta}^c$ for the functional $\cF_c$ that we defined in Section \ref{sec:def-pseudometrics-F}.
In order to show that $\lambda^c = \overline{\delta}^c$, we establish the inequalities in both directions.

\begin{restatable}{lemma}{lemmaalphaf}
\label{lemma:alpha-f}
For every $f$ in $\grammar_c$, there exists $n$ such that for every $x,y$, $ f(x) - f(y) \leq \delta^c_n(x,y)$.
\end{restatable}

This proof is done by induction on the structure of $f$.  A full version can be found in Appendix \ref{app:sec-logics}.

Since $\lambda^c (x,y) = \sup_{f \in \grammar_c} ( f(x) - f(y) )$,  we get the following corollary as an immediate consequence of Lemma \ref{lemma:alpha-f}.
\begin{corollary}
\label{cor:functional-F-geq-logic}
For every discount factor $0< c< 1$,  $\overline{\delta}^c \geq \lambda^c$.
\end{corollary}

We now aim at proving the reverse inequality.
We will use the following result (Lemma A.7.2) from \cite{Ash72} which is also used in~\cite{vanBreugel05a} for discrete-time processes.
\begin{lemma}
\label{lemma:magie}
Let $X$ be a compact Hausdorff space.  Let $A$ be a subset of the set of continuous functions $X \to \mathbb{R}$ such that if $f, g \in A$, then $\max\{ f,g\}$ and $\min \{f,g\}$ are also in $A$.  Consider a function $h$ that can be approximated at each pair of points by functions in $A$, meaning that 
\[ \forall x, y \in X~ \forall \epsilon > 0~ \exists g \in A ~ |h(x) - g(x)| \leq \epsilon \text{ and } |h(y) - g(y)| \leq \epsilon \]
Then $h$ can be approximated by functions in $A$, meaning that $\forall \epsilon > 0 ~ \exists g \in A ~ \forall x \in X~ |h(x) - g(x)| \leq \epsilon$.
\end{lemma}

In order to use this lemma, we need the following lemmas:

\begin{lemma}
\label{lemma:grammar-continuous}
For every $f \in \grammar_c$, the function $x \mapsto f(x)$ is continuous.
\end{lemma}

\begin{proof}
This is done by induction on the structure of $\Lambda_c$. The only case that is not straightforward is when $f =  \langle t \rangle g$. By induction hypothesis, $g$ is continuous.
Since the map $x \mapsto P_t(x)$ is continuous (onto the weak topology), $f = c^t \hat{P}_t g$ is continuous.
\end{proof}

\begin{restatable}{lemma}{lemmaprepmagieF}
\label{lemma:prep-magie-F}
Consider a continuous function $h: E \to [0,1]$ and two states $z, z'$ such that there exists $f$ in the logic $\grammar_c$ such that $|h(z) - h(z')| \leq |f(z) - f(z')|$. Then for every $\epsilon > 0$, there exists $g \in \grammar_c$  such that $|h(z) - g(z)| \leq 2 \epsilon$ and $|h(z') - g(z')| \leq 2 \epsilon$.
\end{restatable}

\begin{proof}
WLOG $h(z) \geq h(z')$ and $f(z) \geq f(z')$ (otherwise consider $1-f$ instead of $f$).
Pick $p,q,r \in \mathbb{Q} \cap [0,1]$ such that
\begin{align*}
p & \in [f(z') - \epsilon, f(z')],\\
q & \in [h(z) - h(z') - \epsilon, h(z) - h(z')],\\
r & \in [h(z'), h(z') + \epsilon].
\end{align*}
Define $g = ( \min \{ f \ominus p, q \} ) \oplus r$. The computations showing that the result follows can be found in Appendix \ref{app:sec-logics}.
\end{proof}

\begin{corollary}
\label{cor:prep-magie}
Consider a continuous function $h: E \to [0,1]$ such that for any two states $z, z'$ there exists $f$ in the logic $\grammar_c$ such that $|h(z) - h(z')| \leq |f(z) - f(z')|$. Then for every compact set $K$ in $E$, there exists a sequence $(g_n)_{n \in \mathbb{N}}$ in $\Lambda_c$ that approximates $h$ on~$K$.
\end{corollary}

\begin{proof}
We have proven in Lemma \ref{lemma:prep-magie-F} that such a function $h$ can be approximated at pairs of states by functions in $\grammar_c$. Now recall that all the functions in $\grammar_c$ are continuous (Lemma \ref{lemma:grammar-continuous}).

We can thus apply Lemma \ref{lemma:magie} on the compact set $K$, and we get that the function $h$ can be approximated by functions in $\grammar_c$.
\end{proof}

\begin{theorem}
\label{thm:lambda-fixpoint-F}
The pseudometric $\lambda^c$ is a fixpoint of $\cF_c$.
\end{theorem}

\begin{proof}
We will omit writing $c$ as an index in this proof.  We already have that $\lambda \leq \cF(\lambda)$ (cf Lemma \ref{lemma:order-m-Fc}), so we only need to prove the reverse direction.

There are countably many expressions in $\grammar$ so we can number them: $\grammar = \{ f_0, f_1, ...\}$. Write $m_k (x,y) = \max_{j \leq k} |f_j(x) - f_j(y)|$. 
Since all the $f_j$ are continuous (Lemma \ref{lemma:grammar-continuous}), the map $m_k$ is also continuous. Furthermore, $m_k \leq m_{k+1}$ and for every two states $x$ and $y$, $\lim_{k \rightarrow \infty} m_k(x,y) = \lambda (x,y)$.

Using Lemma \ref{lemma:wasserstein-limit}, we know that for every states $x,y$ and time $t$, 
\[\sup_k W(m_k)(P_t(x), P_t(y)) = W(\lambda)(P_t(x), P_t(y)).\]
This implies that
\begin{align*}
\cF (\lambda) (x, y) 
&= \sup_{t \geq 0} c^t W(\lambda)(P_t(x), P_t(y)) = \sup_{k} \sup_{t \geq 0} c^t W(m_k)(P_t(x), P_t(y)).
\end{align*}

It is therefore enough to show that for every $k$, every time $t \in \mathbb{Q}_{\geq 0}$ and every pair of states $x,y$, $c^t W(m_k)(P_t(x), P_t(y)) \leq \lambda (x,y)$.
There exists $h \in \cH(m_k)$ such that $W(m_k)(P_t(x), P_t(y)) = \left| \int h~ \dee P_t(x) - \int h~ \dee P_t(y)\right| .$

Since $P_t(x)$ and $P_t(y)$ are tight, for every $\epsilon > 0$, there exists a compact set $K \subset E$ such that $P_t(x, E \setminus K) \leq \epsilon / 4$ and $P_t(y, E \setminus K) \leq \epsilon / 4$.

By Corollary \ref{cor:prep-magie}, there exists $(g_n)_{n \in \mathbb{N}}$ in $\Lambda$ that pointwise converge to $h$ on $K$.  In particular, for $n$ large enough,
\begin{align*}
\left| \int_K g_n ~ \dee P_t(x) - \int_K h ~ \dee P_t(x) \right| \leq \epsilon /4,
\end{align*}
and similarly for $P_t(y)$. We get that for $n$ large enough.
\begin{align*}
& \left| \int_E g_n ~ \dee P_t(x) - \int_E h ~ \dee P_t(x) \right| \\
& \leq \left| \int_K g_n ~ \dee P_t(x) - \int_K h ~ \dee P_t(x) \right| + \int_{E \setminus K} |g_n- h| ~ \dee P_t(x)  \leq \epsilon / 2,
\end{align*}
and similarly for $P_t(y)$. We can thus conclude that
\begin{align*}
& c^t W(m_k) (P_t(x), P_t(y)) = c^t \left| \int_E h ~ \dee P_t(x) - \int_E h ~ \dee P_t(y) \right|\\
& \leq c^t \left| \int_E g_n ~ \dee P_t(x) - \int_E h ~ \dee P_t(x) \right| + c^t \left| \int_E g_n ~ \dee P_t(y) - \int_E h ~ \dee P_t(y) \right|\\
& \qquad + \left| (\langle t \rangle g_n)(x) - ( \langle t \rangle g_n)(y)\right|\\
& \leq \epsilon + \lambda(x, y).
\end{align*}
Since $\epsilon$ is arbitrary,  $c^t W(m_k) (P_t(x), P_t(y)) \leq \lambda (x, y)$. \hfill $\square$
\end{proof}

As a consequence of Theorem \ref{thm:lambda-fixpoint-F} and using Lemma \ref{lemma:comparison-fixpointFc}, we get that $\overline{\delta}^c \leq \lambda^c$. Then with Corollary \ref{cor:functional-F-geq-logic}, we finally get:
\begin{theorem}
\label{thm:metric-equal-logic-functional-F}
The two pseudometrics are equal: $\overline{\delta}^c = \lambda^c$.
\end{theorem}

\section{Obstacles in continuous time}
\label{sec:wrapup}

As we have pointed out throughout this work, although the overall outline is similar to that employed in the discrete case, we are forced to develop new strategies to overcome technical challenges arising from the continuum setting, where the topological properties of the time/state space become crucial elements in the study.

For example, a key obstacle we face in this work is that the fixpoint pseudometric can't be derived directly from the Banach fixed point theorem.  There is no notion of step and we therefore need to consider all times in $\mathbb{R}_{\geq 0}$. In discrete-time, the counterpart of our functional $\cF_c$ is $c$-Lipschitz. However, in our case, this should be thought of as $c^1$. Since we are forced to consider all times and since $\sup_{t \geq 0} c^t = 1$, we cannot find a constant $k < 1$ such that $\cF_c$ is $k$-Lipschitz.  
To overcome this issue, we construct a candidate pseudometric with brute force (we first define a sequence of pseudometrics $\delta_n^c$ and then define our candidate $\overline{\delta}^c$ as their supremum) and then prove that it is indeed the greatest fixpoint.

In addition, the scope of the functional $\cF_c$ also requires careful treatment: for measurability reasons, we have to restrict to the lattice of pseudometrics which generate a subtopology of the pre-existing one. We cannot apply the Kleene fixed point theorem either since this lattice is not complete.
This whole approach toward a fixpoint pseudometric differs substantially from the long-existing methodology.

Furthermore, for some of our key results (e.g., Theorem \ref{thm:lambda-fixpoint-F} and Lemma \ref{lemma:Fc-cont-is-cont}), the proofs rely on the compactness argument to establish certain convergence relations in temporal and/or spatial variables and to further achieve the goals. This type of procedure in general is not required in the discrete setting.

Although from time to time, we do restrict the time variable to rational values thanks to the continuity of the FD-semigroups, this is very different from treating discrete-time models, because rational time stamps cannot be ``ordered'' to represent the notion of ``next step'' in a continuous-time setting. Therefore, we are still working with a true continuous-time dynamics, and it cannot be reduced to a discrete-time problem.

\section{Examples}
\label{sec:example}

\subsection{A toy example}
\label{sec:toy-ex}

Let us consider a process defined on $\{ 0, x, y, z, \partial\}$. Let us first give an intuition for what we are trying to model.  In the states $x$, $y$, $z$, the process is trying to learn a value. In the state $0$, the correct value has been learnt, but in the state $\partial$, an incorrect value has been learnt. From the three `` learning'' states $x, y$ and $z$, the process has very different learning strategies:
\begin{itemize}
\item from $x$,  the process exponentially decays to the correct value represented by the state $0$,
\item from $y$, the process is not even attempting to learn the correct value and thus remains in a learning state, and
\item from $z$, the process slowly learns but it may either learn the correct value ($0$) or an incorrect one ($\partial$).
\end{itemize}
The word ``learning'' is here used only to give colour to the example.

Formally, this process is described by the time-indexed kernels:
\begin{align*}
P_t (x, \{0\}) &= 1 - e^{- \lambda t} & P_t(z, \{ 0\}) &= \frac{1}{2}(1 - e^{- \lambda t}) & P_t(0, \{ 0\}) &= 1\\
P_t (x, \{x\}) &= e^{- \lambda t}  &  P_t(z, \{ \partial\}) &= \frac{1}{2}(1 - e^{- \lambda t}) &P_t(y, \{ y\})& = 1\\
&& P_t(z, \{ z\}) &=e^{- \lambda t} & P_t(\partial, \{ \partial\}) &= 1
\end{align*}
(where $\lambda \geq 0$)
and by the observable function $obs(x) = obs(y) = obs(z) = r \in (0,1)$, $obs(\partial) = 0$ and $obs(0) = 1$.

We will pick as a discount factor $c = e^{-\lambda}$ to simplify notations. We can compute the distance $\overline{\delta}$ (we omit adding $c$ in the notation). 

We will only detail the computations for $\overline{\delta}(0, z)$ in the main section of this paper.  First of all, note that $\delta_0(0, z) = 1 - r$ and that
\begin{align*}
\delta_{n+1} (0, z) 
& = \sup_{t \geq 0} e^{- \lambda t} \left( e^{- \lambda t} \delta_n(0, z) + \frac{1}{2} (1 - e^{- \lambda t}) \delta_n(0, \partial) \right)\\
& = \sup_{0 < \theta \leq 1} \theta \left( \theta \delta_n(0, z) + \frac{1}{2} (1 - \theta) \right)\\
 & =  \sup_{0 < \theta \leq 1} \theta \left[ \frac{1}{2} + \theta \left( \delta_n(0, z) - \frac{1}{2} \right) \right]
\end{align*}
We are thus studying the function $\phi : \theta \mapsto  \theta \left[ \frac{1}{2} + \theta \left( \delta_n(0, z) - \frac{1}{2} \right) \right] $. Its derivative $\phi'$ has value $0$ at $\theta_0 = \frac{1}{2(1 - 2 \delta_n (0, z))}$. There are three distinct cases:
\begin{enumerate}
\item If $0 < \delta_n(0, z) \leq \frac{1}{4}$, in that case,  $0 < \theta_0 \leq 1$ and $\sup_{1 < \theta \leq 1} \phi (\theta)$ is attained in $\theta_0$ and in this case, we have that $\delta_{n+1} (0, z) =  \frac{1}{8(1 - 2 \delta_n (0, z))} \leq \frac{1}{4}$. This means in particular that if $1 - r \leq \frac{1}{4}$, then $\overline{\delta}(0, z) = \frac{1}{4}$.
\item If $\frac{1}{4} \leq \delta_n(0, z) \leq \frac{1}{2}$, then $\theta_0 \geq 1$ and $\phi$ is increasing on $(- \infty, \theta_0]$. In that case, we therefore have that $\sup_{1 < \theta \leq 1} \phi (\theta)$ is attained in $1$ and therefore $\delta_{n+1}(0, z) = \delta_n(0, z)$. 
\item  If $\frac{1}{2} \leq \delta_n(0, z)$, then $\theta_0 \leq 0$ and $\phi$ is increasing on $[\theta_0, + \infty)$. In that case, we therefore have that $\sup_{1 < \theta \leq 1} \phi (\theta)$ is attained in $1$ and therefore $\delta_{n+1}(0, z) = \delta_n(0, z)$. 
\end{enumerate}
We therefore have that
\[ \overline{\delta}(0, z) = \begin{cases}
\frac{1}{4} \qquad \text{if } r \geq \frac{3}{4}\\
1 - r \qquad \text{otherwise.}
\end{cases} \]
%
%
%
%
%
%
%

The other cases can be found in Appendix  \ref{app:example} but we summarize them in a table here. Note that the computation of $\overline{\delta}(x, z)$ is too involved and we therefore only provide an interval.
{ \small \[\begin{NiceArray}{|c|c|c|c|c|c|}
\hline
   & x & y & z & \partial & 0 \\ 
\hline
x & 0 & \frac{1 - r}{2} & \in \left[ \frac{1}{8}, \frac{1}{4} \right]  &  \begin{cases}
r ~~ \text{if } r \geq \frac{1}{2}\\
\frac{1}{2} ~~ \text{otherwise}
\end{cases} & 1 - r \\ \hline
y & \frac{1 - r}{2} & 0 & \frac{1}{4} & r & 1 - r \\ \hline
z &  \in \left[ \frac{1}{8}, \frac{1}{4} \right] & \frac{1}{4}& 0 & \begin{cases}
r ~~  \text{if } r \geq \frac{1}{4}\\
\frac{1}{4} ~~ \text{otherwise}
\end{cases} & \begin{cases}
\frac{1}{4} ~~ \text{if } r \geq \frac{3}{4}\\
1 - r ~~ \text{otherwise}
\end{cases} \\ \hline
\partial &  \begin{cases}
r ~~ \text{if } r \geq \frac{1}{2}\\
\frac{1}{2} ~~ \text{otherwise}
\end{cases} & r & \begin{cases}
r ~~  \text{if } r \geq \frac{1}{4}\\
\frac{1}{4} ~~ \text{otherwise}
\end{cases} & 0 & 1 \\ \hline
0 & 1 - r & 1 -r & \begin{cases}
\frac{1}{4} ~~ \text{if } r \geq \frac{3}{4}\\
1 - r ~~ \text{otherwise}
\end{cases} & 1 & 0 \\
\hline
\end{NiceArray} \] }

Note that even though the process behaves vastly differently from $x$ than from $y$,  we have that $\overline{\delta}(x, 0) = \overline{\delta}(y, 0) $, even though for $t > 0$, we have that $\hat{P}_t obs(x)  > \hat{P}_t obs(y)$. However note that $x$ and $y$ have different distances to other states.


This also happens  in the discrete-time setting and for continuous-time Markov chains. A continuous-time Markov chain is a continuous-time type of processes but where the evolution still occurs as steps. They can be described as jump processes over continuous time.  They have been studied in \cite{Gupta06} by considering the whole trace starting from a single state. It is possible to adapt our work to traces (called trajectories in \cite{Rogers00a}) by adding some additional regularity conditions on the processes; this can be found in \cite{Chen24}. However one should consider the added complexity. For instance, for Brownian motion the kernels $P_t$ are well-known and easy to describe with a density function but that is not the case for the probability measures on trajectories.

\subsection{Two examples based on Brownian motion}
\label{sec:BM-ex}

Previous example is a very simple example which emphasizes some of the difficulties of computing our metric. It may then be tempting to think that our approach cannot yield any result when applied to the real world.  The next examples show that we can still provide some meaningful results when looking at real-life processes such as Brownian motion.
We refer the reader to Section \ref{app:BM} for some background on Brownian motion and hitting times.

\paragraph*{First example:}
We denote $(B_t^x)_{t \geq 0}$ the standard Brownian motion on the real line starting from $x$. We can then define its first hitting time of $0$ or $1$: $\tau:=\inf\{t\geq0:B_{t}^{x}=0\text{ or }B_{t}^{x}=1\}$. 

Our state space is the interval $[0,1]$. For every $x\in[0,1]$ and every $t\geq0$,
let $P_{t}\left(x,\cdot\right)$ be distribution of $B_{t\wedge\tau}^{x}$. In
other words, $P_{t}\left(x,\cdot\right)$ is the distribution of Brownian
motion starting from $x$ running until hitting either 0 or 1 and
getting trapped upon hitting a boundary.  We equip the state space $[0,1]$ with $obs$ defined as $obs(x) = x$. Then, $obs\left(B_{t\wedge\tau}^{x}\right)=B_{t\wedge\tau}^{x}$
and hence for every $x\in\left(0,1\right)$
\begin{align*}
\delta_{1}\left(0,x\right) & =\sup_{t\geq0}c^{t}W\left(\delta_{0}\right)\left(P_{t}\left(0\right),P_{t}\left(x\right)\right) =\sup_{t\geq0}c^{t}\mathbb{E}[B_{t\wedge\tau}^{x}]\\
 & =\sup_{t\geq0}c^{t}\cdot x\text{ (because }B_{t\wedge\tau}^{x}\text{ is a martingale)}\\
 & =x.
\end{align*}
where the first equality follows from the fact that $P_{t}\left(0\right)$ is the dirac distribution $\mathfrak{d}_{0}$ at $0$ and the second equality comes from the fact that any coupling $\gamma\in\Gamma\left(\mathfrak{d}_{0},P_{t}\left(x\right)\right)$ is reduced to the marginal $P_{t}\left(x\right)$. Since $\delta_{0}\left(0,x\right) = \delta_{1}\left(0,x\right)$, we then have $\bar{\delta}\left(0,x\right)=x$.

Similarly, $\bar{\delta}\left(1,y\right)=1-y$
for every $y\in\left[0,1\right]$. It is difficult to compute $\bar{\delta}\left(x,y\right)$
for general $x,y\in\left[0,1\right]$ though.

\paragraph*{Second example:}
This example relies on stochastic differential equations (SDE). We refer the reader to \cite{Oksendal13} for a comprehensive introduction.

Let the state space and $obs$ be the same as above. Let $Q_{t}\left(x,\cdot\right)$
be the distribution of the solution to the SDE
\[
dX_{t}=X_{t}dB^0_{t}+\frac{1}{2}X_{t}dt\text{ with }X_{0}=x.
\]
It can be verified that the solution to this equation is the process $X_{t}=xe^{B^0_{t}}$. In this case, we
also have $Q_{t}\left(0\right)=\mathfrak{d}_{0}$. Again,
for every $x\in\left[0,1\right]$,
\begin{align*}
\delta_{1}\left(0,x\right) & =\sup_{t\geq0}c^{t}W\left(\delta_{0}\right)\left(Q_{t}\left(0\right),Q_{t}\left(x\right)\right) =\sup_{t\geq0}c^{t}\mathbb{E}[obs(xe^{B^0_{t}})]\\
 & =\sup_{t\geq0}c^{t}\left(x\mathbb{E}[e^{B^0_{t}} ~|~ t\leq\tau']+\mathbb{P}\left(t>\tau'\right)\right),
\end{align*}
where $\tau':=\inf\left\{ s\geq0 ~|~ B^0_{s}\geq-\ln x\right\} $ and $\mathbb{E}$ and $\mathbb{P}$ denote expected values and probabilities for the standard Brownian motion starting in $0$. The distribution
of $\tau'$, as well as the joint distribution of $\left(B^0_{t},\tau'\right)$,
has been determined explicitly (see Chapter 2, Section 8 in \cite{Karatzas12} for instance). Even if we only consider the second
term in the expression above, we have
\[
\delta_{1}\left(0,x\right)\geq\sup_{t\geq0}c^{t}\mathbb{P}\left(t>\tau'\right)=\sup_{t\geq0}c^{t}\frac{2}{\sqrt{2\pi}}\int_{\frac{-\ln x}{\sqrt{t}}}^{\infty}e^{-\frac{y^{2}}{2}}dy.
\]
It is possible for the right hand side above
to be greater than $x$. For example, if $x=\frac{1}{e}$, then by
choosing $t=4$, we have that the right hand side above is no smaller
than $c\frac{2}{\sqrt{2\pi}}\int_{\frac{1}{2}}^{\infty}e^{-\frac{y^{2}}{2}}dy$,
which will be (strictly) greater than $\frac{1}{e}$ (i.e., $x$)
provided that $c$ is sufficiently close to 1. As a consequence, for
this example, we have $\bar{\delta}\left(0,x\right)>x$. 

These two examples demonstrate different behaviors of the two systems,
while the first system ``maintains'' the mass at the starting point
(expectation is constant $x$), the second system dissipates the mass
to the right (which is the direction of larger values of $obs$).
Therefore, processes starting from the same point $x$ possess different
distances to the static path between these two systems. 

A very important observation from these examples based on Brownian motion is that even though we cannot explicitly compute values, we are still able to compare state behaviours.

\section{Conclusion}
\label{sec:conclusion}

In our previous work \cite{Chen19a,Chen20}, we showed that we needed to use trajectories in order to define a meaningful notion of behavioural equivalence.  However working with trajectories is extremely complex as notions do not translate easily from states to trajectories; for instance, we said that a measurable set $B$ of trajectories is time-$obs$-closed if for every two trajectories $\omega, \omega'$ such that for every time $t$, $obs(\omega(t)) = obs(\omega'(t))$, then $\omega \in B$ if, and only if $\omega' \in B$. The $\sigma$-algebra of all time-$obs$-closed sets cannot be simply described. 

To explain why this present work does not deal with trajectories, we need to first discuss the example that lead us to trajectories in \cite{Chen19a}: consider Brownian motion on the real line equipped with the function $obs = \mathds{1}_{\{ 0\}}$. There are four $obs$-closed sets: $\emptyset, \{ 0\}, \mathbb{R} \setminus \{ 0\}, \mathbb{R}$ and for any $x \neq 0$ and any time $t$, $P_t(x, \{ 0\}) = 0$. This meant that one could not distinguish between the states $1$ and $1000$ in this specific case. Using trajectories enabled us to consider intervals of time. In this current work, we have decided to instead ``smooth'' the function $obs$ so as to prevent singling points out without needing to deal with trajectories.

Let us go back to our examples.  As shown in those examples, computing $\overline{\delta}$ is quite an involved process. It would have been interesting to adapt our example in Section \ref{sec:toy-ex} to the real-line with $obs(x) = e^{- x^2}$ and consider other processes such as Brownian motion which stops once it hits $0$ or the Ornstein–Uhlenbeck process (a variation of Brownian motion which is ``attracted'' to 0). Having to deal with transport theory, a supremum (over time) and the inductive definition of $\overline{\delta}$ makes it virtually impossible to compute any harder example.  As we have seen in the second example (Section \ref{sec:BM-ex}), it may still be possible to compare the behaviours of two states by stating for instance that ``the behaviour of the process starting from $x$ is closer to that starting from $y$ than from $z$''. The problem of transport theory and the Kantorovich metric also exists in discrete time and there are interesting ways to deal with it, for instance through the MICo distance \cite{Castro21}. 

One of the advantages of optimal transport is that the Kantorovich duality gives us a way of computing bounds. As one notes however, there is again some difficulty with having a supremum over time in the definition of our functional $\cF_c$. In particular, we can only provide lower bounds for $\overline{\delta}$.

One avenue of work on this could be to study replacements for this supremum, such as integrals over time for instance. Note that the real-valued logic would also need to be adapted even though it seems to generalize discrete-time really well.

This work is, as far as we know, the first behavioural metric adapted to the continuous-time case. Clearly much remains to be done, particularly the exploration of examples and connexions to broader classes of processes, such as for example those defined by stochastic differential equations.

\paragraph*{Acknowledgments.}
Linan Chen and Prakash Panangaden were supported by NSERC discovery grants

Florence Clerc was supported by the EPSRC-UKRI grant EP/Y000455/1, by a CREATE grant for the project INTER-MATH-AI, by NSERC and by IVADO through the DEEL Project CRDPJ 537462-18.

\paragraph*{Disclosure of Interests.}
The authors have no competing interests to declare.

\bibliographystyle{abbrv}
\bibliography{../../MFPS2019/main}

\appendix

\section{Bisimulation metric for discrete-time Labelled Markov Processes}
\label{app:DT-metrics}

Let us recall the existing work on discrete-time Labelled Markov Processes adapted to our framework. For the detailed version,  we refer the readers to the works by Desharnais et al. \cite{Desharnais02a,Desharnais04,Panangaden09} and van Breugel et al. \cite{vanBreugel05a}.

A labelled Markov Process is a family of Markov sub-kernels indexed by a set of what are called actions. Our work can easily be adapted to this framework by adding indices. For the sake of readability, we will not consider actions in this work and that is why we will only present the discrete-time work for a single Markov kernel.

Consider a Markov kernel $\tau$ on a Polish space $(E, \cE)$. Note that $\tau$ is a Markov kernel and in particular for every $x \in E$, $\tau(x, E) = 1$.

Define the functional $F$ on 1-bounded pseudometrics on $E$ by
\[ F(m) (x, y) = cW(m) (\tau(x), \tau(y)),\]
where $0< c< 1$ is a discount factor.
This functional has a fixpoint which is a bisimulation metric and is denoted $d_\cC$ in \cite{vanBreugel05a}. It is also characterized by the logic defined inductively:
\[ \phi := 1 ~|~ \diamond \phi ~|~ \min ( \phi, \phi) ~|~ 1 - \phi ~|~ \phi \ominus q, \]
which is interpreted as $(\diamond \phi)(x) = c \int \phi ~\dee \tau(x)$ etc.  Indeed, this logic defines a pseudometric on $E$ by $d_\cL (x, y) = \sup_{\phi} | \phi (x) - \phi(y)|$ and the two pseudometrics coincide: $d_\cL = c d_\cC$.

\section{Proofs for Section \ref{sec:background}}
\label{app:proof-sec-background}

  \subsection{Lower semi-continuity}
  \label{app:semi-cont}

\begin{restatable}{lemma}{lemmasuplowersemicontinuous}
\label{lemma:sup-lower-semicontinuous}
Assume there is an arbitrary family of continuous functions $f_i : X \to \mathbb{R}$ ($i \in \cI$) and define $f(x) = \sup_{i \in \cI} f_i(x)$ for every $x \in X$.  Then $f$ is lower semi-continuous.
\end{restatable}

\begin{proof}
Pick $r \in \mathbb{R}$.  Then
\begin{align*}
f^{-1}((r, + \infty))
& = \{ x ~|~ \sup_{i \in \cI} f_i(x) > r\}\\
& = \{ x ~|~ \exists i \in \cI ~ f_i(x) > r\} = \bigcup_{i \in \cI} f_i^{-1} ((r, + \infty)).
\end{align*}
Since $f_i$ is continuous, the set $f_i^{-1}((r + \infty))$ is open and therefore $f^{-1}((r, + \infty))$ is open which concludes the proof. \hfill $\square$
\end{proof}

The converse is also true, as Baire's theorem states:
\begin{theorem}
\label{thm:baire}
If $X$ is a metric space and if a function $f: X \to \mathbb{R}$ is lower semi-continuous,  then $f$ is the limit of an increasing sequence of real-valued continuous functions on $X$.
\end{theorem}

    \subsection{Couplings}

In order to prove Lemma \ref{lemma:couplings-compact}, we will need to prove some more results first.

\begin{proposition}
\label{prop:two}
If $P$ and $Q$ are two measures on a Polish space $X$ such that for every continuous and bounded function
$f: X \to \mathbb{R}$, we have \(\int f ~\dee P = \int f~\dee Q\), then \(P = Q\).  
\end{proposition}

\begin{proof}{}
For any open set $U$, its indicator function $\mathds{1}_U$ is lower semicontinuous, which means that there exists an increasing sequence of continuous functions $f_n$ converging pointwise to $\mathds{1}_U$ (see Theorem \ref{thm:baire}).  Without loss of generality, we can assume that the functions $f_n$ are non negative (consider the sequence $\max \{ 0, f_n\}$ instead if they are not non-negative).  Using the monotone convergence theorem, we know that
\[ \lim_{n \to \infty} \int f_n ~ \dee P = \int \mathds{1}_U ~\dee P = P(U) \]
and similarly for $Q$.  By our hypothesis on $P$ and $Q$, we obtain that for every open set $U$, $P(U) = Q(U)$.

Since $P$ and $Q$ agree on the topology, they agree on the Borel algebra by Caratheodory's extension theorem. \hfill $\square$
\end{proof}

\begin{lemma}
\label{lemma:coupling-closed}
Given two probability measures $P$ and $Q$ on a Polish space $X$, the set of couplings $\Gamma(P, Q)$ is closed in the weak topology.
\end{lemma}
\begin{proof}
Consider a sequence $(\gamma_n)_{n \in \mathbb{N}}$ of couplings of $P$ and $Q$ weakly converging to a measure $\mu$ on $X \times X$, meaning that for every bounded continuous functions $f : X \times X \to \mathbb{R}$, $\lim_{n \to \infty} \int f ~ \dee \gamma_n = \int f~ \dee \mu$.  Note that $\mu$ is indeed a probability measure.  We only have to prove that the marginals of $\mu$ are $P$ and $Q$.  We will only do it for the first one, i.e.  showing that for every measurable set $A$, $\mu(A \times X) = P(A)$ since the case for $Q$ works the same.

Let $\pi:X\times X\to X$ be the first projection map.  The first marginal of $\gamma_n$ and $\mu$ are obtained as the pushforward measures $\pi_* \gamma_n$ and $\pi_*\mu$.  Indeed, $\pi_* \mu$ is defined for all measurable set $A \subset X$ as $\pi_* \mu (A) = \mu (\pi^{-1} (A)) = \mu(A \times X)$ (and similarly for the $\gamma_n$'s).  This means that for every $n$,  $\pi_* \gamma_n = P$.

For an arbitrary continuous bounded function $f: X \to \mathbb{R}$,  define the function $g: X \times X \to \mathbb{R}$ as $g = f \circ \pi$, i.e.  $g(x, y) = f(x)$.  The function $g$ is also continuous and bounded, so by weak convergence of the sequence $(\gamma_n)_{n \in \mathbb{N}}$ to $\mu$, we get that
\[ \lim_{n \to \infty} \int f \circ \pi ~ \dee  \gamma_n = \int f \circ \pi ~ \dee \mu. \]
Using the change of variables formula, we obtain that for any continuous bounded function $f$
\[ \int f ~ \dee P =  \lim_{n \to \infty} \int f ~ \dee  (\pi_*\gamma_n) = \int f ~ \dee  (\pi_*\mu). \]
Using  Proposition \ref{prop:two}, we get that $\pi_*\mu = P$. \hfill $\square$
\end{proof}

\begin{lemma}
\label{lemma:couplings-tight}
Consider two probability measures $P$ and $Q$ on Polish spaces $X$ and $Y$ respectively.   Then the set of couplings $\Gamma (P, Q)$ is tight: for all $\epsilon > 0$, there
exists a compact set $K$ in $X \times Y$ such that for all coupling $\gamma \in \Gamma(P, Q)$, $\gamma(K) > 1- \epsilon$
\end{lemma}

\begin{proof}
First note that by Ulam's tightness theorem, the probability measures $P$ and $Q$ are tight.

Consider $\epsilon > 0$.  Since $P$ and $Q$ are tight, there
exist two compact sets $K$ and $K'$ such that $P(X \setminus K) \leq \epsilon/2$ and $Q(Y \setminus K') \leq \epsilon/2$.  Define the set $C = K \times K'$.  This set is compact as a product of two compact sets.  Furthermore, for every $\gamma \in \Gamma(P,Q)$,
\begin{align*}
 \gamma((X \times Y) \setminus C) 
& = \gamma \left( [(X \setminus K) \times Y] \cup [X \times ( Y \setminus K')] \right)\\
& \leq \gamma[(X \setminus K) \times Y] + \gamma [X \times ( Y \setminus K')]\\
& = P(X \setminus K) + Q(Y \setminus K') \leq \frac{\epsilon}{2} + \frac{\epsilon}{2} = \epsilon.
\end{align*}
This shows that the set $\Gamma(P, Q)$ is tight.\hfill $\square$
\end{proof}

We can now prove Lemma \ref{lemma:couplings-compact}. Let us start by restating it.

\lemmacouplingscompact*

\begin{proof}
Using Lemma \ref{lemma:couplings-tight}, we know that the set $\Gamma(P, Q)$ is tight.
Applying Prokhorov's theorem (see Theorem \ref{thm:prokhorov}),  we get that the set $\Gamma(P, Q)$ is precompact.

Since the set $\Gamma(P, Q)$ is also closed (see Lemma \ref{lemma:coupling-closed}), then it is compact.\hfill $\square$
\end{proof}

We have used Prokhorov's theorem in the previous proof.  Here is the version cited in \cite{Villani08}
\begin{theorem}
\label{thm:prokhorov}
Given a Polish space $\cX$, a subset $\cP$ of the set of probabilities on $\cX$ is precompact for the weak topology if and only if it is tight.
\end{theorem}

   \subsection{Optimal transport theory}
\label{app:transport}

\begin{remark}
\label{rem:dual-pseudometric}
The expression that we provided for the Kantorovich duality is not the exact expression in Theorem 5.10 of \cite{Villani08} but the one found in Particular Case 5.16.  The former expression also applies to our case.  Indeed, according to Theorem 5.10, the dual expression is
\[
\min_{\gamma \in \Gamma(\mu, \nu)} \int c ~ \dee \gamma
= \max_{\phi, \psi} \left( \int \phi ~ \dee  \mu - \int \psi ~ \dee \nu \right) \]
where $\phi$ and $\psi$ are such that  $\forall x, y~ |\phi(x) - \psi(y)| \leq c(x,y)$.
However, since $c(x,x) = 0$ for every $x \in \cX$, then for any pair of functions $\phi$ and $\psi$ considered, for all $x, \in X$,
$ |\phi(x) - \psi(x)| \leq c(x,x) = 0,$ 
which implies that $\phi = \psi$.
\end{remark}

\lemmaoptimaltransportpseudometric*

\begin{proof}
The first expression $W(c)(\mu, \nu) =  \min_{\gamma \in \Gamma(\mu, \nu)} \int c ~ \dee \gamma$ immediately gives us that $W(c)$ is 1-bounded (since we are working on the space of probability distributions on $\cX$) and that $W(c)$ is symmetric (by the change of variable formula with the function on $\cX \times \cX$, $s(x,y) = (y,x)$).

The second expression $W(c)(\mu, \nu) = \max_{h \in \cH(c)} \left| \int h ~ \dee \mu - \int h ~\dee \nu  \right|$ immediately gives us that $W(c)(\mu, \mu) = 0$ and for any three probability distributions $\mu, \nu, \theta$ and for any function $h \in \cH(c)$,
\begin{align*}
\left| \int h ~ \dee \mu - \int h ~\dee \theta \right|
& \leq \left| \int h ~ \dee \mu - \int h ~\dee \nu \right| + \left| \int h ~ \dee \nu - \int h ~\dee \theta \right|\\
& \leq W(c) (\mu, \nu) + W(c)(\nu, \theta).
\end{align*}
Since this holds for every $h \in \cH(c)$, we get the triangular inequality. \hfill $\square$
\end{proof}

\lemmawassersteinlimit*

\begin{proof}
For each $c_k$, the optimal transport cost $W(c_k)(P, Q)$ is attained for a coupling $\pi_k$.  Using Lemma \ref{lemma:couplings-compact}, we know that the space of couplings $\Gamma(P, Q)$ is compact. We can thus extract a subsequence that we will still denote by $(\pi_k)_{k \in \mathbb{N}}$ which converges weakly to some coupling $\pi \in \Gamma(P, Q)$. We will show that this coupling $\pi$ is in fact an optimal transference plan.

By Monotone Convergence Theorem,
\[ \int c ~ \dee \pi = \lim_{k \rightarrow \infty } \int c_k ~ \dee \pi. \]
Consider $\epsilon > 0$. There exists $k$ such that
\begin{align}
\label{eq:numberg}
\int c ~ \dee \pi & \leq \epsilon  \int c_k ~ \dee \pi
\end{align}

Since $\pi_k$ converges weakly to $\pi$ and since $c_k$ is continuous and bounded,
\[ \int c_k ~ \dee \pi = \lim_{n \rightarrow \infty} \int c_k ~ \dee \pi_n,\]
which implies that there exists $n_k \geq k$ such that
\begin{align}
\label{eq:numberh}
\int c_k ~ \dee \pi & \leq \epsilon +  \int c_k ~ \dee \pi_{n_k}.
\end{align}
Putting Equations (\ref{eq:numberg}) and (\ref{eq:numberh}) together, we get
\begin{align*}
\int c ~ \dee \pi
& \leq 2 \epsilon +  \int c_k ~ \dee \pi_{n_k}\\
&  \leq 2 \epsilon +  \int c_{n_k} ~ \dee \pi_{n_k} \quad \text{ since $(c_n)_{n \in \mathbb{N}}$ is increasing}\\
 & = 2 \epsilon + W(c_{n_k})(P, Q).
\end{align*}
This implies that
\[ \int c ~ \dee \pi \leq \lim_{k \rightarrow \infty} W(c_k)(P, Q). \]

The other inequality is trivial since $c_k \leq c$. \hfill $\square$
\end{proof}

\section{Details for Section \ref{sec:CTsystems}}
\label{app:CTsystems}

   \subsection{Feller-Dynkin processes}

Under the conditions of FD-semigroups, strong continuity is equivalent to the apparently weaker condition (see Lemma III.  6.7 in \cite{Rogers00a} for the proof):
\[ \forall f \in C_0(E)~~ \forall x\in E, ~~ \lim_{t\downarrow 0} (\hat{P}_t f) (x) = f(x)\]

The authors also offer the following useful extension (Theorem III.6.1):
\begin{theorem}
\label{thm:Riesz}
A bounded linear functional $\phi$ on $C_0(E)$ may be written uniquely in the form
\[ \phi(f) = \int_E f(x)~ \mu(\dee x) \]
where $\mu$ is a signed measure on $E$ of finite total variation.
\end{theorem}
As stated in the Riesz representation theorem, the above measure $\mu$ is inner regular.   Theorem \ref{thm:Riesz} is also known as the Riesz-Markov-Kakutani representation theorem in some other references.

Theorem \ref{thm:Riesz} has the following corollary (Theorem III.6.2) where $b\mathcal{E}$ denotes the set of bounded, $\cE$-measurable functions $E \to \mathbb{R}$.
\begin{restatable}{corollary}{cormarkovfromsemigroup}
\label{cor:markov-from-semigroup}
Suppose that $V: C_0(E) \to b\mathcal{E}$ is a (bounded) linear operator that is sub-Markov in the sense that $0 \leq f \leq 1$ implies $0 \leq Vf \leq 1$.  Then there exists a unique sub-Markov kernel (also denoted by) $V$ on $(E, \mathcal{E})$ such that for all $f \in C_0(E)$ and $x \in E$
\[ Vf(x) = \int f(y) ~ V(x, \dee y). \]
\end{restatable}

While the proof is left as an exercise in \cite{Rogers00a}, we choose to explicitly write it down for clarity


\begin{proof}
For every $x$ in $E$, write $V_x$ for the functional $V_x(f) = Vf(x)$.  This functional is bounded and linear which enables us to use Theorem \ref{thm:Riesz}: there exists a signed measure $\mu_x$ on $E$ of finite total variation such that
\[ V_x (f) = \int_E f(y) ~\mu_x(\dee y) \]
We claim that $V: (x, B) \mapsto \mu_x(B)$ is a sub-Markov kernel, i.e.
\begin{itemize}
\item For all $x$ in $E$, $V(x, -)$ is a subprobability measure on $(E, \mathcal{E})$
\item For all $B$ in $\mathcal{E}$, $V( -, B)$ is $\mathcal{E}$-measurable.
\end{itemize}

The first condition directly follows from the definition of $V$: $V(x, -) = \mu_x$ which is a measure and furthermore $V(x, E) = \mu_x(E) = V_x (1)$ (where $1$ is the constant function over the whole space $E$ which value is $1$).  Using the hypothesis that $V$ is Markov, we get that $V 1 \leq 1$.  We have to be more careful in order to prove that $V(x, B) \geq 0$ for every measurable set $B$: this is a consequence of the regularity of the measure $\mu_x$ (see Proposition 11 of section 21.4 of \cite{Royden10}).  This shows that $\mu_x$ is a subprobability measure on $(E, \mathcal{E})$.

Now, we have to prove that for every $B \in \mathcal{E}$, $V(-, B)$ is measurable.  Recall that the set $E$ is $\sigma$-compact: there exists countably many compact sets $K_k$ such that $E = \bigcup_{k \in \mathbb{N}} K_k$.  For $n \in \mathbb{N}$, define $E_n =  \bigcup_{k = 0}^n K_k$ and $B_n = B \cap E_n$.

 For every $n \in \mathbb{N}$, there exists a sequence of  functions $(f_j^n)_{j \in \mathbb{N}} \subset C_0(E)$  that converges pointwise to $\mathds{1}_{B_n}$, i.e.  for every $x \in E$, $\mu_x(B_n) = \lim_{j \to +\infty} Vf_j^n(x)$.  Since the operator $V: C_0(E) \to b\cE$, the maps $Vf_j^n$ are measurable which means that $V(-, B_n): x \mapsto \mu_x(B_n)$ is measurable.

Since for every $x \in E$, $\mu_x(B) = \lim_{n \to +\infty} \mu_x(B_n)$, this further means that $V(-, B): x \mapsto \mu_x(B)$ is measurable. \hfill $\square$
\end{proof}

We can then use this corollary to derive Proposition \ref{prop:Riesz-use} which relates these FD-semigroups with Markov
kernels.  This allows one to see the connection with familiar
probabilistic transition systems.  

\section{Some background on Brownian motion}
\label{app:BM}

   \subsection{Definition of Brownian motion}

Brownian motion is a stochastic process describing the irregular motion of
a particle being buffeted by invisible molecules.  Now its range of
applicability extends far beyond its initial application~\cite{Karatzas12}.
The following definition is from~\cite{Karatzas12}.
\begin{definition}
A standard one-dimensional Brownian motion is a Markov process:
\[ B = (B_t,\cF_t), 0 \leq t < \infty  \]
where $(\cF_t)_{t \geq 0}$ is the natural filtration adapted to $(B_t)_{t \geq 0}$ and whe stochastic process $(B_t)_{t \geq 0}$ is
defined on a probability space $(\Omega,\cF,P)$ with the properties
\begin{enumerate}
\item $B_0 = 0$ almost surely,
\item for $0\leq s < t$, $B_t - B_s$ is independent of $\cF_s$ and is
  normally distributed with mean $0$ and variance $t-s$.
\end{enumerate}
\end{definition}

In this very special process, one can start at any place, there is an
overall translation symmetry which makes calculations more tractable. We denote $(B_t^x)_{t \geq 0}$ the standard Brownian motion on the real line starting from $x$ in section \ref{sec:BM-ex}.

In order to do any calculations we use the following fundamental formula:  If
the process is at $x$ at time $0$ then at time $t$ the probability that it
is in the (measurable) set $D$ is given by
\[ P_t(x,D) = \int_{y\in D} \frac{1}{\sqrt{2\pi t}}
  \exp\left(-\frac{(x-y)^2}{2t}\right)\mathrm{d}y. \]
The associated FD semigroup is the following: for $f \in C_0(\mathbb{R})$ and $x \in \mathbb{R}$,
\[ \hat{P}_t (f) (x) = \int_{y} \frac{f(y)}{\sqrt{2\pi t}}
  \exp\left(-\frac{(x-y)^2}{2t}\right)\mathrm{d}y. \]

   \subsection{Hitting time for stochastic processes}

We follow the definitions of Karatzas and Shreve in~\cite{Karatzas12}.

Assume that $(X_t)_{t \geq 0}$ is a stochastic process on $(\Omega, \cF)$ such that $(X_t)_{t \geq 0}$ takes values in a state space $E$ equipped with its Borel algebra $\cB (\cO))$, has right-continuous paths (for every time $t$, $\omega(t) = \lim_{s \to t, s > t} \omega (s)$) and is adapted to a filtration $(\cF_t)_{t \geq 0}$.

\begin{definition}
Given a measurable set $C \in \cB(\cO)$, the \emph{hitting time} is the random time
\[ T_C (\omega) = \inf \{ t \geq 0 ~|~ X_t(\omega) \in C \}. \]
\end{definition}
Intuitively $T_C(\omega)$ corresponds to the first time the trajectory $\omega$ ``touches'' the set $C$. 

   \subsection{Hitting time for Brownian motion}

Consider $(B_t^x)_{t \geq 0}$ the standard Brownian motion on the real line starting from $x > 0$. With probability 1, the trajectories of Brownian motion are continuous, which means that
\[ \mathbb{P}^x (H_{(- \infty, 0]} < t) = \mathbb{P}^x (H_{\{0\}} < t).  \]
Here we denote $\mathbb{P}^x$ for the probability on trajectories for Brownian motion starting from $x$. We refer the reader to our previous papers \cite{Chen19a,Chen20,Chen23} for a full description.

A standard result is that
\[ \mathbb{P}^x (H_{\{0\}} < t) = \sqrt{\frac{2}{\pi}} \int^{+ \infty}_{x / \sqrt t} e^{- s^2 / 2} ds. \]
We refer the reader to \cite{Borodin97} for formulas regarding various hitting times of standard Brownian motion and many of its variants.

In section \ref{sec:BM-ex}, we introduced the hitting time $\tau:=\inf\{t\geq0:B_{t}^{x}=0\text{ or }B_{t}^{x}=1\}$, with $x \in [0,1]$. This corresponds to the hitting time of the set $(- \infty, 0] \cup [1, + \infty)$.  We can then define the process $B_{t\wedge\tau}^{x}$ on the state space $[0,1]$: the intuition is that this process behaves like Brownian motion until it reaches one of the ``boundaries'' (either 0 or 1) where it becomes ``stuck''.

Note that there are two main types of boundaries that one can consider for Brownian motion on the real line: when the process hits the boundary, it can either ``get stuck'' / vanish (boundary with absorption) or bounce back (boundary with reflection). 

\section{Proofs for Section \ref{sec:gen-CT}}
\label{app:gen-CT}

\lemmaordermFc*

\begin{proof}
Consider a pair of states $x,y$. Then
\begin{align*}
\cF_c(m)(x,y)
& = \sup_{t \geq 0} c^t W(m)(P_t(x), P_t(y))\\
& \geq W(m)(P_0(x), P_0(y))\\
& = \inf_{\gamma \in \Gamma(P_0(x), P_0(y))} \int m ~ \dee \gamma.
\end{align*}
Since $P_0(x)$ is the dirac distribution at $x$ and similarly for $P_0(y)$, the only coupling $\gamma$ between $P_0(x)$ and $P_0(y)$ is the product measure $P_0(x) \times P_0(y)$ and thus $W(m)(P_0(x), P_0(y)) = m(x, y)$, which concludes the proof. \hfill $\square$
\end{proof}

\thmmfixpointF*

\begin{proof}
We will omit writing $c$ as an index for the pseudometrics $\delta_n$ and $\overline{\delta}$ throughout this proof. Fix two states $x,y$ and a time $t$.

The space of finite measures on $E \times E$ is a linear topological space.
Using Lemma \ref{lemma:couplings-compact}, we know that the set of couplings $\Gamma (P_t(x), P_t(y))$ is a compact subset.  It is also convex.

The space of bounded pseudometrics on $E$ is also a linear topological space.  We have defined a sequence $\delta_0, \delta_1, ...$ on that space.  We define $Y$ to be the set of linear combinations of pseudometrics $\sum_{n \in \mathbb{N}} a_n \delta_n$ such that for every $n$, $a_n \geq 0$ (and finitely many are non-zero) and $\sum_{n \in \mathbb{N}} a_n = 1$.  This set $Y$ is convex.

Define the function
\begin{align*}
\Xi : \Gamma (P_t(x), P_t(y)) \times Y & \to [0,1]\\
(\gamma, m) & \mapsto\int m ~ \dee \gamma.
\end{align*}

For $\gamma \in \Gamma (P_t(x), P_t(y))$, the map $\Xi(\gamma, \cdot)$ is continuous by the dominated convergence theorem and it is monotone and hence quasiconcave.  For a given $m \in Y$, by definition of the Lebesgue integral, $\Xi(\cdot, m)$ is continuous and linear.

We can therefore apply Sion's minimax theorem:
\begin{align}
\label{eq:minimax-F}
\min_{\gamma \in \Gamma (P_t(x), P_t(y))} \sup_{m \in Y} \int m ~ \dee \gamma =  \sup_{m \in Y} \min_{\gamma \in \Gamma (P_t(x), P_t(y))}  \int m ~ \dee \gamma.
\end{align}

Now note that for an arbitrary functional $\Psi: Y \to \mathbb{R}$ such that for any two pseudometrics $m$ and $m'$,  $m \leq m' \Rightarrow \Psi(m) \leq \Psi(m')$, 
\begin{align}
\label{eq:numberd}
\sup_{m \in Y} \Psi(m) = \sup_{n \in \mathbb{N}} \Psi (\delta_n).
\end{align}

Indeed since $\delta_n \in Y$ for every $n$,  we have that $\sup_{m \in Y} \Psi(m) \geq \sup_{n \in \mathbb{N}} \Psi (\delta_n)$.  For the other direction, note that for every $m \in Y$, there exists $n$ such that for all $k \geq n$, $a_k = 0$.  Thus $m \leq \delta_n$ and therefore $\Psi(m) \leq \Psi(\delta_n)$ which is enough to prove the other direction.

The right-hand side of Equation (\ref{eq:minimax-F}) is
\begin{align*}
\sup_{m \in Y} \min_{\gamma \in \Gamma (P_t(x), P_t(y))}  \int m ~ \dee \gamma 
&= \sup_{m \in Y} W(m)(P_t(x), P_t(y))\\
& = \sup_{n \in \mathbb{N}} W(\delta_n)(P_t(x), P_t(y))\quad \text{(previous result in Equation (\ref{eq:numberd})).}
\end{align*}

The left-hand side of Equation (\ref{eq:minimax-F}) is
\begin{align*}
& \min_{\gamma \in  \Gamma (P_t(x), P_t(y))} \sup_{m \in Y} \int m ~ \dee \gamma\\
& = \min_{\gamma \in  \Gamma (P_t(x), P_t(y))} \sup_{n \in \mathbb{N}} \int \delta_n ~ \dee \gamma \quad \text{(previous result in Equation (\ref{eq:numberd}))}\\
& = \min_{\gamma \in  \Gamma (P_t(x), P_t(y))}  \int \sup_{n \in \mathbb{N}} \delta_n ~ \dee \gamma \quad \text{(dominated convergence theorem)}\\
& = \min_{\gamma \in  \Gamma (P_t(x), P_t(y))}  \int \overline{\delta} ~ \dee \gamma \\
& = W(\overline{\delta})(P_t(x), P_t(y)).
\end{align*}

We thus have that
\begin{align*}
\cF_c(\overline{\delta}) (x, y) & = \sup_{t \geq 0} c^t W(\overline{\delta})(P_t(x), P_t(y))\\
& = \sup_{t \geq 0} c^t \sup_{n \in \mathbb{N}} W(\delta_n)(P_t(x), P_t(y)) \quad \text{using Sion's minimax theorem}\\
& = \sup_{n \in \mathbb{N}}  \sup_{t \geq 0} c^t W(\delta_n)(P_t(x), P_t(y))\\
& = \sup_{n \in \mathbb{N}}  \cF_c(\delta_n)(x,y)\\
& =  \sup_{n \in \mathbb{N}} \delta_{n+1} (x,y) \quad \text{by definition of } (\delta_n)_{n \in \mathbb{N}}\\
& = \overline{\delta}(x,y) \quad \text{by definition of } \overline{\delta}.
\end{align*} \hfill $\square$
\end{proof}

\section{Proofs for Section \ref{sec:logics}}
\label{app:sec-logics}

\lemmaalphaf*

\begin{proof}
This proof is done by induction on the structure of $f$:
\begin{itemize}
\item If $f = q$, then $f(x) - f(y) = 0 \leq \delta^c_0(x,y)$.
\item If $f = obs$, then $f(x) - f(y) = obs(x) - obs(y) \leq  \delta^c_0(x,y)$.
\item If $f = 1 - g$ and there exists $n$ such that for every $x,y$, $ g(x) - g(y) \leq \delta^c_n(x,y)$, then $f(x) - f(y) = g(y) - g(x) \leq \delta^c_n(y,x) =  \delta^c_n(x,y)$.
\item If $f = g \ominus q$  and there exists $n$ such that for every $x,y$, $ g(x) - g(y) \leq \delta^c_n(x,y)$,  then it is enough to study the case when $f(x) = g(x) - q \geq 0$ and $f(y) = 0 \geq g(y) -q$. In that case,
\begin{align*}
f(x) - f(y) & = g(x) - q \leq g(x) - q - (g(y) - q) = g(x) - g(y) \leq \delta_n^c(x,y)\\
f(y) - f(x) &=  q - g(x) \leq 0 \leq \delta^c_n(x,y).
\end{align*}
\item If $f = \min \{ f_1, f_2\}$ and for $i = 1,2$ there exists $n_i$ such that for every $x,y$, $ f_i(x) - f_i(y) \leq \delta^c_{n_i}(x,y)$. Then write $n = \max\{ n_1, n_2\}$. There really is only one case to consider: $f(x) = f_1(x) \leq f_2(x)$ and $f(y) = f_2(y) \leq f_1(y)$. In that case
\[ f(x) - f(y) = f_1(x) - f_2(y) \leq f_2(x) - f_2(y) \leq \delta^c_{n_2}(x,y) \leq \delta^c_n(x,y). \]
\item If $f = \langle t \rangle g$ and there exists $n$ such that for every $x,y$, $ g(x) - g(y) \leq \delta^c_n(x,y)$,  then
\begin{align*}
f(x) - f(y)
& = c^t \hat{P}_t g(x) - c^t \hat{P}_t g(y) = c^t \left( \hat{P}_t g(x) -\hat{P}_t g(y) \right)\\
& \leq c^t W(\delta_n^c)(P_t(x), P_t(y))\\
& \leq \delta_{n+1}^c (x,y).
\end{align*}
\end{itemize} \hfill $\square$
\end{proof}

\lemmaprepmagieF*

\begin{proof}
WLOG $h(z) \geq h(z')$ and $f(z) \geq f(z')$ (otherwise consider $1-f$ instead of $f$).

Pick $p,q,r \in \mathbb{Q} \cap [0,1]$ such that
\begin{align*}
p & \in [f(z') - \epsilon, f(z')],\\
q & \in [h(z) - h(z') - \epsilon, h(z) - h(z')],\\
r & \in [h(z'), h(z') + \epsilon].
\end{align*}
Define $g = ( \min \{ f \ominus p, q \} ) \oplus r$. Then,
\begin{align*}
(f \ominus p)(z)
& \in [f(z) - f(z'), f(z) - f(z') + \epsilon] ~~ \text{since } f(z) \geq f(z'), \\
(\min \{ f \ominus p, q \})(z)
& = q ~~ \text{since } q \leq h(z) - h(z') \leq f(z) - f(z'),\\
g(z) & = \min \{ 1, q + r\} \in [h(z) - \epsilon, h(z) + \epsilon)] ~~ \text{as } h(z) \leq 1,
\end{align*}
which means that $|h(z) - g(z)| \leq  \epsilon$.
\begin{align*}
(f \ominus p)(z')
& = \max \{ 0, f(z') - p\} \in [0, \epsilon],\\
(\min \{ f \ominus p, q \})(z')
& \in [0, \epsilon],\\
g(z')& \in [h(z'), h(z') +2 \epsilon],
\end{align*}
which means that $|h(z') - g(z')| \leq 2 \epsilon$. \hfill $\square$
\end{proof}

\section{Computations for Section \ref{sec:example}}
\label{app:example}

\begin{itemize}
\item $\overline{\delta}(0, \partial)$: First of all, note that $\delta_0 (0, \partial) = 1$ which means that 
\[1 = \delta_0 (0, \partial) \leq \overline{\delta}(0, \partial) \leq 1\]
 and thus $\overline{\delta}(0, \partial) = 1$.
 \item $\overline{\delta}(0, x)$: First of all, $\delta_0 (0, x) = 1 - r$. We can now compute $\delta_1 (0, x):$
\begin{align*}
 \delta_1 (0, x) & = \sup_{t \geq 0} e^{-\lambda t} \left( ( 1 - e^{- \lambda t}) \times \delta_0(0,0) +  e^{- \lambda t} \times \delta_0(0, x) \right) \\
 &= \sup_{t \geq 0} e^{-2 \lambda t} \times(1 - r) = 1-r 
\end{align*}
Since $\delta_1 (0, x) = \delta_0 (0, x) = 1 -r$, we also have that $\overline{\delta}(0, x) = 1 - r$.
\item $\overline{\delta}(0, y)$: First of all, $\delta_0 (0, y) = 1 - r$. We can now compute $\delta_1 (0, y)$:
\begin{align*}
 \delta_1 (0,y) & = \sup_{t \geq 0} e^{-\lambda t} \delta_0 (0,y) \\
 &= \sup_{t \geq 0} e^{- \lambda t} \times(1 - r) = 1-r 
\end{align*}
Since $\delta_1(0, y) = \delta_0(0, y) = 1 - r$, we also have that $\overline{\delta}(0, y) = 1 - r$.
\item $\overline{\delta}(y, \partial)$: First of all, $\delta_0 (y, \partial) = r$. We can now compute $\delta_1 (y, \partial)$:
\begin{align*}
 \delta_1 (y, \partial) & = \sup_{t \geq 0} e^{-\lambda t} \delta_0 (y, \partial) \\
 &= \sup_{t \geq 0} e^{- \lambda t} r = r
\end{align*}
Since $\delta_1 (y, \partial) = \delta_0(y, \partial) = r$, we also have that $\overline{\delta}(y, \partial) = r$.

\item $\overline{\delta}(x, y)$: First, by induction on $n$, we prove that $\delta_n(x,y) = \epsilon_n (1 - r)$ where the sequence $(\epsilon_n)_{n \in \mathbb{N}}$ is defined as
\[ \epsilon_0 = 0 \qquad \epsilon_{n + 1} = \frac{1}{4( 1 - \epsilon_n)}\]
and satisfies $0< \epsilon_n \leq \frac{1}{2}$ for every $n \in \mathbb{N}$.
Indeed, $\delta_0(x, y) = r - r = 0$ and assuming the result holds at rank $n$:
\begin{align*}
\delta_{n+1} (x, y) & = \sup_{t \geq 0} e^{- \lambda t} \left( (1 - e^{- \lambda t} ) \delta_n (0, y) + e^{- \lambda t}  \delta_n(x, y) \right)\\
& = \sup_{0< \theta \leq 1} \theta \left( (1 - \theta )(1 - r) + \theta (1 - r)\epsilon_n \right)\\
& = (1 - r) \sup_{0< \theta \leq 1} \theta \left( 1 - \theta(1 - \epsilon_n) \right).
\end{align*}
We are now left looking for the supremum of the function $\phi: \theta \mapsto  \theta \left( 1 - \theta(1 - \epsilon_n) \right)$ on $(0,1]$.  We have that
\[ \phi'(\theta) =  1 - 2 \theta(1 - \epsilon_n) \]
and we can thus conclude that its maximum is attained at $\frac{1}{2 (1 - \epsilon_n)} \in \left(\frac{1}{2}, 1\right]$ and thus
\[ \delta_{n+1} (x, y) = (1 - r) \phi \left( \frac{1}{2 (1 - \epsilon_n)} \right) = \frac{1}{2 (1 - \epsilon_n)} \times \frac{1}{2} \in \left( \frac{1}{4}, \frac{1}{2} \right]. \]

We are now left to study the limit of the sequence $(\epsilon_n)_{n \in \mathbb{N}}$. We know that this limit exists so we only have to compute it.
This limit $l$ satisfies $l =  \frac{1}{4( 1 - l)}$
and we can thus conclude that $\overline{\delta}(x,y) = \frac{1 - r}{2}$.

\item $\overline{\delta}(y, z)$: Similarly to what we did before, we obtain that $\delta_0 = 0$ and
\[ \delta_{n+1} (y ,z) = \sup_{0 < \theta \leq 1} \theta \left( \frac{1 - \theta}{2} + \theta \delta_n(y, z) \right) \leq \frac{1}{4}. \]
This supremum is attained for $\theta = \frac{1}{2(1 - 2\delta_n(y, z))} \leq 1$ and thus
\[ \delta_{n+1} (y ,z) =\frac{1}{8(1 - 2\delta_n(y, z))} \]
and finally $\overline{\delta}(y, z) = \frac{1}{4}$.

\item $\overline{\delta}(x, \partial)$: First, we have that $ \delta_0(x, \partial) = r - 0 = r$ and for every $n \geq 0$,
\begin{align*}
\delta_{n+1} (x, \partial) & = \sup_{t \geq 0} e^{- \lambda t} \left( (1 - e^{- \lambda t} ) \delta_n (0, \partial) + e^{- \lambda t}  \delta_n(x, \partial) \right)\\
& = \sup_{0< \theta \leq 1} \theta \left( (1 - \theta ) + \theta  \delta_n(x, \partial) \right)
\end{align*}
We are now left looking for the supremum of the function 
\[\phi: \theta \mapsto \theta \left( (1 - \theta ) + \theta  \delta_n(x, \partial) \right) \text{ on }(0,1].\] 
Its derivative $\phi'$ has value $0$ at $\theta_0 = \frac{1}{2(1 - \delta_n (x, \partial))}$.  From the theory, we know that $\delta_n(x, \partial) \leq 1$ and we therefore know that $\phi$ reaches its supremum in $\theta_0 \geq \frac{1}{2}$. There are now two cases to consider: either $\theta_0 > 1$ or $\theta_0 \leq 1$. We get that
\[ \overline{\delta}(x, \partial) = \begin{cases}
r \qquad \text{if } r \geq \frac{1}{2}\\
\frac{1}{2} \qquad \text{otherwise}.
\end{cases} \]

\item $\overline{\delta}(z, \partial)$: This case is similar to the previous ones: $ \delta_0(z, \partial) = r$ and 
\begin{align*}
\delta_{n+1} (x, \partial) 
& = \sup_{0< \theta \leq 1} \theta \left( \frac{1 - \theta }{2} + \theta  \delta_n(z, \partial) \right)
\end{align*}
Similarly to the case of $\overline{\delta}(0, z)$. we have three cases to study, and we end up with
\[ \overline{\delta}(z, \partial)= \begin{cases}
r \qquad  \text{if } r \geq \frac{1}{4}\\
\frac{1}{4} \qquad \text{otherwise}.
\end{cases} \]

\item $\overline{\delta}(x, z)$: due to a conflict of notation, we will write $\mathfrak{d}_b$ for the Dirac measure centered in $b$.

First, we note that $\overline{\delta}(x, z) \leq \frac{1}{2}$. Similarly to what we have done so far, we have that $\delta_0(x, z) = 0$ and
\[ \delta_{n+1}(x, z) = \sup_{0 < \theta \leq 1} \theta W (\delta_n) (\theta \mathfrak{d}_x + (1 - \theta) \mathfrak{d}_0, \theta \mathfrak{d}_z + \frac{1- \theta}{2} (\mathfrak{d}_0 + \mathfrak{d}_\partial)). \]
We can now show by induction that for every $n$, $\delta_n(x, z) \leq \frac{1}{4}$: consider the coupling $\gamma_\theta$ defined by
\[ \gamma(0,0)\footnote{Due to how heavy the notations are already, we write $\gamma(0,0)$ instead of $\gamma(\{ (0, 0)\})$} = \gamma (0, \partial) = \frac{1- \theta}{2} \text{ and } \gamma (x, z) = \theta. \]
We then have that
\begin{align*}
& W (\delta_n) \left(\theta \mathfrak{d}_x + (1 - \theta) \mathfrak{d}_0, \theta \mathfrak{d}_z + \frac{1- \theta}{2} (\mathfrak{d}_0 + \mathfrak{d}_\partial)\right)\\
& \leq \int \delta_n ~ \dee \gamma_\theta\\
& = \frac{1- \theta}{2} \delta_n(0, \partial) + \theta \delta_n(x, z) \\
& = \frac{1}{2} - \theta \left( \frac{1}{2} - \delta_n(x, z) \right)
\end{align*}
This means that
\begin{align*}
\delta_{n +1} (x, z) & \leq  \sup_{0 < \theta \leq 1} \theta \left[  \frac{1}{2} - \theta \left( \frac{1}{2} - \delta_n(x, z) \right) \right]
\end{align*}
And we get that $\delta_{n +1} (x, z) \leq \frac{1}{4}$ as usual.

We also have that
\begin{align*}
\overline{\delta}(x, z)
& \geq \sup_{0 < \theta \leq 1} \theta \left| \left( \theta obs(x) + (1 - \theta) obs (0) \right) - \left( \theta obs(z) + \frac{1 - \theta}{2} (obs(\partial) + obs (0)) \right) \right|\\
& = \sup_{0 < \theta \leq 1} \theta \left| \left( \theta r + (1 - \theta)  \right) - \left( \theta r + \frac{1 - \theta}{2} \right) \right|\\
& =  \frac{1}{2}\sup_{0 < \theta \leq 1} \theta (1 - \theta) = \frac{1}{8}
\end{align*}

From this we conclude that $ \frac{1}{8} \leq \overline{\delta}(x, z) \leq  \frac{1}{4}$.

\end{itemize}

\end{document}